\documentclass[11pt]{article}

\usepackage[margin=1in]{geometry}
\usepackage{amsmath,amssymb,amsthm}
\usepackage{mathtools}
\usepackage{bm}
\usepackage{graphicx}
\usepackage{booktabs}
\usepackage{natbib}
\usepackage{setspace}
\usepackage{caption}
\usepackage[hidelinks]{hyperref}

\bibpunct{(}{)}{;}{a}{}{,}
\DeclareMathOperator{\Ga}{Gamma}
\DeclareMathOperator{\GIG}{GIG}
\DeclareMathOperator{\IG}{IG}
\DeclareMathOperator{\diag}{diag}

\newcommand{\R}{\mathbb{R}}
\newcommand{\E}{\mathbb{E}}
\newcommand{\Prob}{\mathbb{P}}

\theoremstyle{plain}
\newtheorem{theorem}{Theorem}
\newtheorem{proposition}{Proposition}
\newtheorem{lemma}{Lemma}

\theoremstyle{definition}
\newtheorem{condition}{Condition}
\newtheorem{remark}{Remark}
\newtheorem{algorithm}{Algorithm}

\begin{document}

\begin{center}
{\Large\bfseries Nonparanormal Bayesian Learning of Directed Acyclic Graphs\\[0.25em]
under Gamma and Inverse-Gamma Innovation Priors:\\[0.25em]
Closed-Form Scores and Informed Sampling}\\[1.3em]

{\large Samaneh~Nazari$^{a}$,\quad Mohammad~Arashi$^{a,\ast}$}\\[0.8em]

\small
$^{a}$Department of Statistics, Faculty of Mathematical Sciences,\\
Ferdowsi University of Mashhad, Mashhad, Iran\\[0.5em]
$^{\ast}$Corresponding author: \texttt{arashi@um.ac.ir}\\[0.6em]
\end{center}

\vspace{0.6em}

\begin{abstract}
\noindent\textbf{Summary.}
Bayesian structure learning for directed acyclic graphs (DAGs) is a central
tool for reconstructing biological signalling and regulatory networks, yet the
methodology has been developed almost exclusively under the assumption that the
data are jointly Gaussian. In the proteomic and flow-cytometry applications that motivate the methodology this assumption is routinely violated, i.e., measured concentrations are strongly right-skewed and heavy-tailed, and a Gaussian DAG fitted to such data recovers spurious or misdirected edges. We develop a fully Bayesian framework for DAG learning in the nonparanormal family, which replaces Gaussianity by the much weaker requirement that unknown strictly increasing transformations of the margins are jointly Gaussian. Working on the modified Cholesky parameterization of the latent precision matrix, we introduce two innovation-variance priors that have not previously been combined with the nonparanormal construction, a non-conjugate Normal--Gamma prior, which decouples coefficient shrinkage from variance regularization, and a conjugate Normal--Inverse-Gamma prior. For both priors we obtain the node-wise marginal likelihood in closed form---through a modified Bessel function of the third kind for the Gamma prior and through a Student-$t$ form for the Inverse-Gamma prior---so that every move of a Markov chain Monte Carlo sampler over DAG space can be scored without numerical integration. Exploiting these closed-form scores we build a locally-balanced informed sampler whose Metropolis--Hastings acceptance probability collapses to a ratio of neighbourhood normalizers, and a Bessel-free score that scales the sampler to hundreds of nodes. On simulated data the nonparanormal samplers match Gaussian methods when the data are Gaussian and dominate them sharply when the margins are skewed. On a human T-cell protein-signalling
data they recover well-established interactions from the literature consensus
network at high posterior probability and clearly outperform constraint-based
competitors, performing comparably to a Gaussian Bayesian model.
\end{abstract}

\vspace{0.4em}

\vspace{0.8em}

\section{Introduction}
\label{sec:intro}

Directed acyclic graph (DAG) models, also called Bayesian networks, encode the
conditional independence structure of a multivariate distribution through a
directed graph and have become a standard formalism for reconstructing causal
mechanisms in the biomedical sciences \citep{spirtes2000,pearl2009,sachs2005}.
In a flow-cytometry experiment, for example, the simultaneous abundances of a
panel of signalling proteins are recorded in many thousands of individual cells,
and the scientific objective is to recover the directed pathway through which a
perturbation of one protein propagates to the others. When the joint law of the
variables is a Gaussian DAG, the precision matrix factorizes through the
modified Cholesky decomposition \citep{pourahmadi1999}, the graph is encoded by
the support of the Cholesky factor, and a substantial Bayesian literature
provides priors, marginal likelihoods, and samplers for learning that support
\citep{benDavid2015,cao2019,castelletti2020,castelletti2022bcdag}.

The Gaussian assumption is, however, an awkward fit to the data that motivate
the methodology. Figure~\ref{fig:motivation} displays four representative
margins from the human T-cell protein-signalling data of
\citet{sachs2005}, where every margin is sharply right-skewed, with sample skewness between two and four and excess kurtosis ranging from roughly ten to thirty, and the Gaussian density fitted by moment matching (solid curve) misrepresents both the mode and the tail. The normal quantile--quantile plots in the lower row bend systematically away from the diagonal. A Gaussian DAG fitted directly to such data must absorb this marginal misspecification into the dependence structure, and the result is well documented, e.g., spurious edges between marginally heavy-tailed variables, and reversed edge orientations driven by skewness rather than conditional dependence \citep{harris2013,liu2009}. Discretizing the data, the historically common
alternative, discards quantitative information and introduces sensitivity to bin
boundaries.

\begin{figure}[t]
\centering
\includegraphics[width=\textwidth]{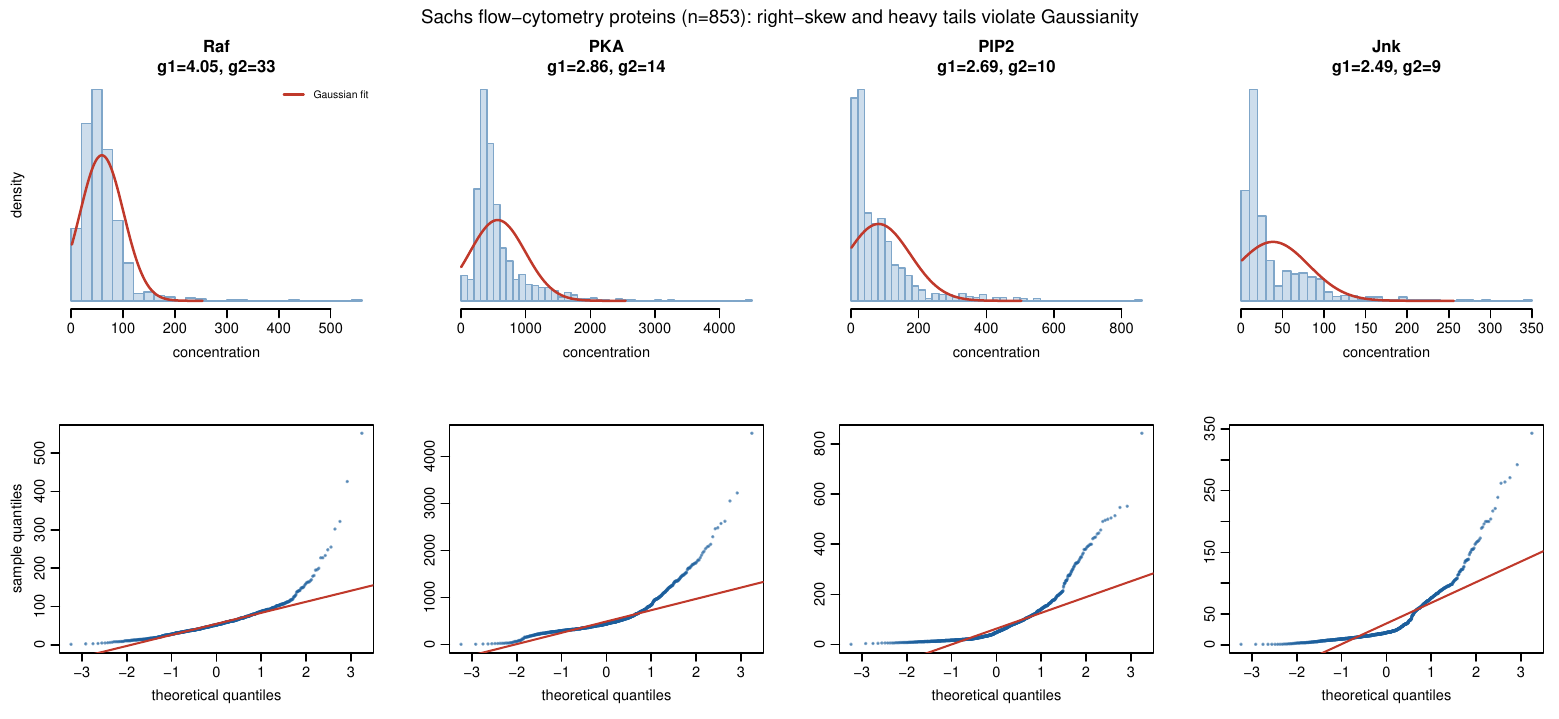}
\caption{Marginal distributions of four signalling proteins from the human T-cell
flow-cytometry data of \citet{sachs2005} ($n=853$). Top row:
histograms with the moment-matched Gaussian density (solid); $\widehat\gamma_1$
and $\widehat\gamma_2$ denote sample skewness and excess kurtosis. Bottom row:
normal quantile--quantile plots. The pronounced right-skew and heavy tails make
a joint Gaussian model untenable and motivate a semiparametric relaxation.}
\label{fig:motivation}
\end{figure}

The nonparanormal family of \citet{liu2009} offers exactly the relaxation
required. A random vector $\bm X=(X_1,\dots,X_p)^\top$ is nonparanormal if there
exist strictly increasing transformations $f_1,\dots,f_p$ such that
$\bm Z=(f_1(X_1),\dots,f_p(X_p))^\top$ is multivariate Gaussian; equivalently,
the dependence is governed by a Gaussian copula while the margins are
unrestricted. The family is large enough to accommodate the skewness and heavy
tails of Figure~\ref{fig:motivation} yet retains the entire conditional
independence calculus of the Gaussian law, because monotone marginal
transformations leave the graph invariant. Figure~\ref{fig:schematic}
illustrates the construction; skewed observed margins (a) are mapped by unknown
monotone transforms (b) to a latent Gaussian vector whose precision matrix
factorizes along a DAG (c). In the undirected setting the nonparanormal has been thoroughly developed, both frequentist \citep{liu2009,liu2012,xue2012,harris2013} and, more recently, Bayesian; but for directed graphs, where the modified Cholesky factor carries the causal
ordering, a Bayesian treatment is essentially absent. This is the gap the
present paper addresses.

\begin{figure}[t]
\centering
\includegraphics[width=\textwidth]{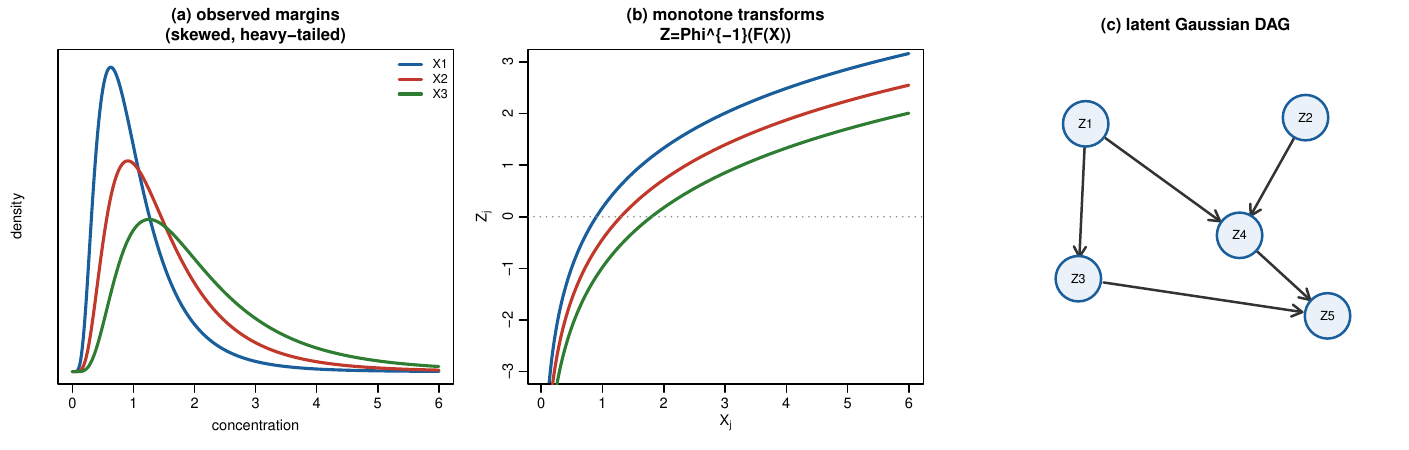}
\caption{The nonparanormal DAG model. Observed margins (a) are skewed and
heavy-tailed; unknown strictly increasing transforms
$f_j=\Phi^{-1}\!\circ F_j$ (b) map them to a latent Gaussian vector (c) whose
precision matrix admits the modified Cholesky factorization
$\Sigma^{-1}=(I-L)^\top D^{-1}(I-L)$, with the support of $L$ encoding the DAG.}
\label{fig:schematic}
\end{figure}

Two obstacles have kept Bayesian DAG learning Gaussian. The first is the
marginal transforms: in a fully Bayesian treatment they are infinite-dimensional
nuisance parameters. We dispose of them by the normal-score device of
\citet{liu2009}: each margin is mapped through its Winsorized empirical
distribution function and the standard normal quantile function, producing
pseudo-Gaussian variables on which the Gaussian DAG machinery operates. The
price is a transform-estimation error, and a central part of our analysis is a
propagation bound (Theorem~\ref{thm:propagation}) showing that this error enters
the node scores only at the nonparanormal rate $\sqrt{\log p/n}$, so that the
selection guarantees of the Gaussian model are inherited. Figure~\ref{fig:transform}
shows the transform mapping a strongly skewed margin to an approximately standard
normal one.

\begin{figure}[t]
\centering
\includegraphics[width=0.82\textwidth]{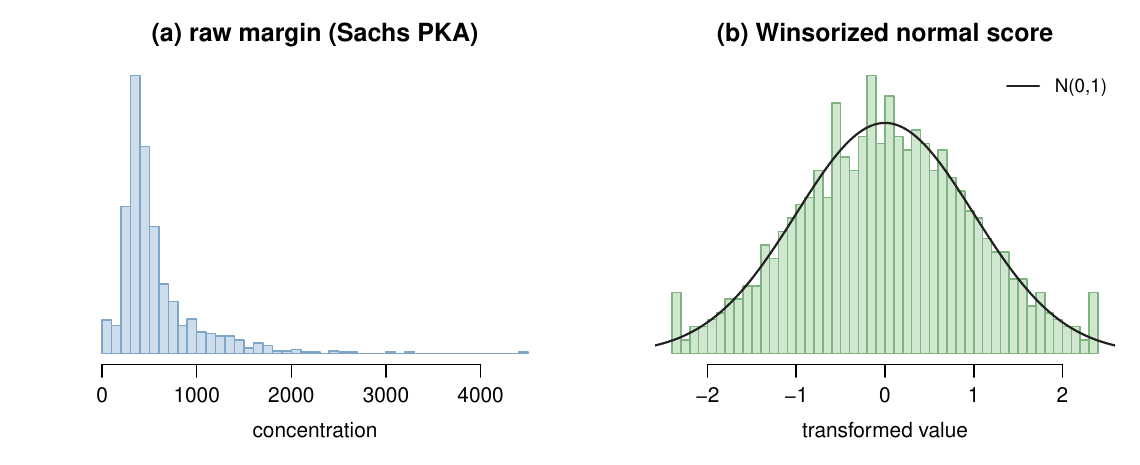}
\caption{Effect of the Winsorized normal-score transform on a skewed margin.
(a) Raw protein concentrations; (b) the transformed variable
$\widehat Z_j=\Phi^{-1}\{\widetilde F_j(X_j)\}$ overlaid with the $N(0,1)$
density. The transform removes marginal skewness while preserving the rank
information that determines the latent dependence.}
\label{fig:transform}
\end{figure}

The second obstacle is the prior on the latent Gaussian DAG parameters. The
conjugate DAG-Wishart prior \citep{benDavid2015,cao2019} couples the shrinkage
applied to the Cholesky coefficients to the regularization of the innovation
variances through a single shape parameter, which is restrictive when, as in
sparse network recovery, one wishes to shrink coefficients aggressively while
leaving the variances comparatively free. We therefore study two innovation
priors on the modified-Cholesky parameters $(L,D)$ that decouple these two
roles. The first places independent Gamma priors on the innovation variances
$D_{jj}$ and conditionally Gaussian priors on the Cholesky coefficients; we call
the resulting model the nonparanormal Normal--Gamma model (NPN-NG). The second
places independent Inverse-Gamma priors on the variances, giving the
nonparanormal Normal--Inverse-Gamma model (NPN-NIG). Neither combination has
appeared in the nonparanormal DAG setting, and the two priors behave differently
in a way we make precise: the Gamma prior shrinks weakly supported edges more
aggressively, controlling false discoveries, whereas the Inverse-Gamma prior is
conjugate and computationally lighter.

A recurring concern about non-conjugate priors is that they forfeit the
closed-form marginal likelihoods on which efficient samplers over DAG space
depend. We show that this concern is unfounded here. Integrating the Cholesky
coefficients and then the innovation variance yields, for the Gamma prior, a
node score expressed through the modified Bessel function of the third kind
(Theorem~\ref{thm:scoreNG}), and for the Inverse-Gamma prior a Student-$t$ score
(Theorem~\ref{thm:scoreNIG}). Because every neighbour of the current DAG can
thus be scored cheaply, we can do better than a random-walk sampler: we
construct a locally-balanced informed proposal \citep{zanella2020} whose
acceptance probability collapses to a ratio of neighbourhood normalizers
(Proposition~\ref{prop:informed}) and which improves the effective sample size
per second by several-fold (between two and nine times in our experiments). To
reach hundreds of nodes we
replace the Bessel evaluation by a uniform asymptotic approximation with relative
error $O(n^{-2})$ and bound the resulting perturbation of the stationary
distribution (Proposition~\ref{prop:bessfree}).

The contributions of the paper are as follows. (i) We formulate Bayesian
nonparanormal DAG learning and combine it with Gamma and Inverse-Gamma
innovation priors, neither of which has previously been used in this setting.
(ii) We characterize the differential shrinkage of the two
priors through the posterior mean of the innovation precision, obtaining an
explicit quadratic characterization of which prior shrinks more
(Theorem~\ref{thm:shrinkage}). (iii) We prove a propagation
bound transferring the nonparanormal transform rate to the scores
(Theorem~\ref{thm:propagation}) and, building on it, strong graph-selection
consistency under $\log p=o(n)$ (Theorem~\ref{thm:selection}). (iv) We give an
informed, scalable sampler with reversibility and optimality guarantees. 

Section~\ref{sec:model} sets up the nonparanormal DAG model and the transform.
Section~\ref{sec:prior} introduces the two priors. Section~\ref{sec:scores}
derives the closed-form scores and establishes order-invariance.
Section~\ref{sec:comp} develops the full conditionals, the informed sampler, and
the scalable score. Section~\ref{sec:theory} states the shrinkage,
propagation, and selection-consistency theorems. Sections~\ref{sec:sim}
and~\ref{sec:real} report the simulation study and the Sachs application, and
Section~\ref{sec:disc} concludes. All proofs are collected in the Appendix.

\section{The nonparanormal DAG model}
\label{sec:model}


Let $\bm X=(X_1,\dots,X_p)^\top$ be the observed vector and suppose it is
nonparanormal: there exist strictly increasing differentiable functions
$f_1,\dots,f_p:\R\to\R$ such that
\begin{eqnarray}
\bm Z=(f_1(X_1),\dots,f_p(X_p))^\top\sim N_p(\bm 0,\Sigma),
\label{eq:npn}
\end{eqnarray}
where, for identifiability of the transforms, $\Sigma$ is a correlation matrix
\citep{liu2009}. Conditional independence among the $X_j$ coincides with
conditional independence among the latent $Z_j$, because the transforms act
coordinatewise and are monotone; the graphical structure is therefore a property
of $\Sigma^{-1}$ alone.

Fix a topological ordering of the variables and assume, after relabelling, that
it is $1,2,\dots,p$, so that a candidate DAG $\mathcal D$ has edges only from
lower- to higher-indexed nodes. Write $\mathrm{pa}_{\mathcal D}(j)\subseteq
\{1,\dots,j-1\}$ for the parent set of node $j$. The modified Cholesky
decomposition \citep{pourahmadi1999} writes
\begin{eqnarray}
\Sigma^{-1}=(I-L)^\top D^{-1}(I-L),
\label{eq:cholesky}
\end{eqnarray}
where $L=(L_{jk})$ is strictly lower triangular, with $L_{jk}\ne 0$ permitted
only when $k\in\mathrm{pa}_{\mathcal D}(j)$, and $D=\diag(D_{11},\dots,D_{pp})$
collects the innovation variances. Equation~\eqref{eq:cholesky} is equivalent to
the system of Gaussian structural equations
\begin{eqnarray}
Z_j=\sum_{k\in\mathrm{pa}_{\mathcal D}(j)} L_{jk}\,Z_k+\varepsilon_j,
\qquad \varepsilon_j\sim N(0,D_{jj}),
\qquad j=1,\dots,p,
\label{eq:sem}
\end{eqnarray}
with mutually independent innovations. The likelihood thus factorizes node-wise,
the $j$th factor being a Gaussian regression of $Z_j$ on its parents. This
factorization, inherited from the latent Gaussian law, is what makes the
modified Cholesky parameterization so convenient: learning the DAG reduces to
$p$ coupled variable-selection problems, one per node.

\subsection{The normal-score transform}
\label{sec:transform}

The transforms $f_j$ are unknown. Following \citet{liu2009} we estimate each by
the Winsorized normal-score map. Given a sample $X_{1j},\dots,X_{nj}$ with
empirical distribution function $\widehat F_j$, define the truncation level
\begin{eqnarray}
\delta_n=\frac{1}{4\,n^{1/4}\sqrt{\pi\log n}}
\label{eq:delta}
\end{eqnarray}
and the Winsorized empirical distribution function
$\widetilde F_j(x)=\min\{1-\delta_n,\max(\delta_n,\widehat F_j(x))\}$, and set
\begin{eqnarray}
\widehat Z_{ij}=\Phi^{-1}\{\widetilde F_j(X_{ij})\},\qquad i=1,\dots,n,\;
j=1,\dots,p,
\label{eq:nsc}
\end{eqnarray}
where $\Phi$ is the standard normal distribution function. The truncation in
\eqref{eq:delta} prevents the boundary ranks from being mapped to $\pm\infty$ and
is calibrated so that the bias it introduces is of smaller order than the
stochastic error; the choice \eqref{eq:delta} is that of \citet{liu2009}. We
write $\widehat{\bm Z}=(\widehat Z_{ij})\in\R^{n\times p}$ for the matrix of
transformed observations, $\widehat{\bm z}_j$ for its $j$th column, and, for a
parent set $S$, $\widehat{\bm Z}_S\in\R^{n\times|S|}$ for the submatrix of
columns indexed by $S$. All subsequent computation treats $\widehat{\bm Z}$ as
data; Theorem~\ref{thm:propagation} quantifies the cost of doing so.

\section{Normal-Gamma(Inverse-Gamma) priors}
\label{sec:prior}

For node $j$ with parent set $S=\mathrm{pa}_{\mathcal D}(j)$ of size $s=|S|$,
collect the active Cholesky coefficients in $\bm\beta_j=(L_{jk})_{k\in S}
\in\R^{s}$. Conditional on the innovation variance we take a Gaussian
coefficient prior scaled by that variance,
\begin{eqnarray}
\bm\beta_j\mid D_{jj}\sim N_s(\bm 0,\;\tau^2 D_{jj}\,I_s),
\label{eq:betaprior}
\end{eqnarray}
with $\tau^2>0$ a fixed scale. The scaling by $D_{jj}$ renders the conditional
posterior of $\bm\beta_j$ Gaussian with a variance-free mean and is the device
that produces the closed forms of Section~\ref{sec:scores}. The two models
differ only in the prior on $D_{jj}$. The Normal--Gamma model (NPN-NG)
takes
\begin{eqnarray}
D_{jj}\sim\Ga(a,b),\qquad
\pi_{\mathrm{NG}}(D)=\frac{b^{a}}{\Gamma(a)}\,D^{a-1}e^{-bD},
\qquad a,b>0,
\label{eq:ngprior}
\end{eqnarray}
a prior on the variance itself rather than its reciprocal, which makes the model
non-conjugate. The Normal--Inverse-Gamma model (NPN-NIG) takes
\begin{eqnarray}
D_{jj}\sim\IG(a_0,b_0),\qquad
\pi_{\mathrm{NIG}}(D)=\frac{b_0^{a_0}}{\Gamma(a_0)}\,D^{-a_0-1}e^{-b_0/D},
\qquad a_0,b_0>0,
\label{eq:nigprior}
\end{eqnarray}
the conjugate choice. Figure~\ref{fig:priors}(a) contrasts the two densities:
the Gamma prior places no mass arbitrarily close to zero and a light right tail,
whereas the Inverse-Gamma prior has a heavy right tail and vanishing mass at the
origin. These differences propagate to the node scores, displayed in
Figure~\ref{fig:priors}(b) and analysed in Section~\ref{sec:theory}.

\begin{figure}[t]
\centering
\includegraphics[width=0.86\textwidth]{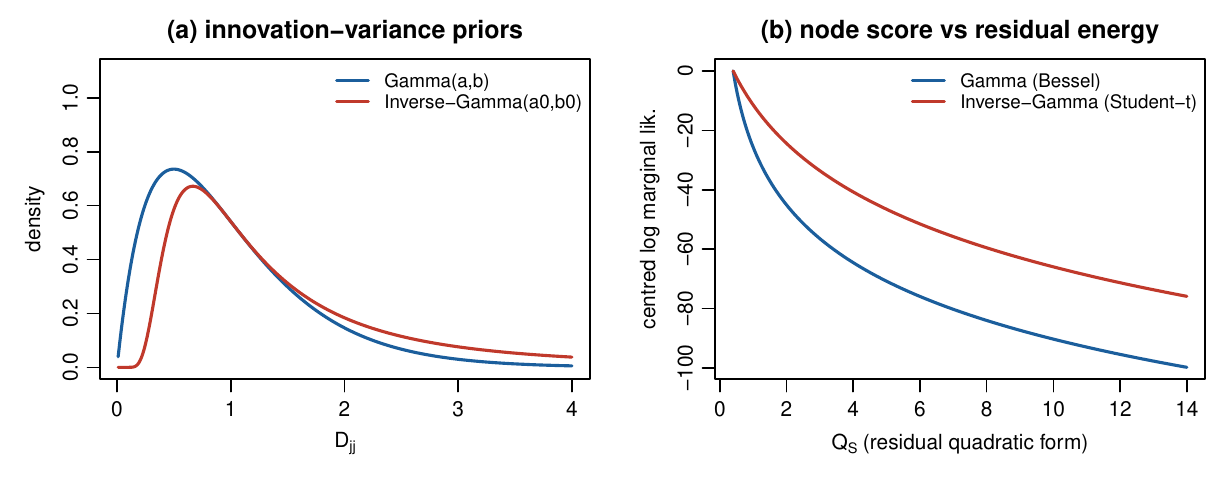}
\caption{(a) The Gamma and Inverse-Gamma innovation-variance priors. (b) The
induced node score (centred log marginal likelihood) as a function of the
residual quadratic form $Q_S$: the Bessel (Gamma) score and the Student-$t$
(Inverse-Gamma) score.}
\label{fig:priors}
\end{figure}

Over DAGs we use the edge prior that treats each admissible edge as an
independent Bernoulli inclusion,
\begin{eqnarray}
\pi(\mathcal D)\propto \theta^{\,|\mathcal D|}(1-\theta)^{\binom{p}{2}-|\mathcal D|},
\qquad 0<\theta<1,
\label{eq:dagprior}
\end{eqnarray}
where $|\mathcal D|$ is the number of edges; equivalently each node $j$ carries
an independent inclusion probability $\theta$ on each of its $j-1$ admissible
parents. For high-dimensional sparsity we allow $\theta=\theta_n\to0$ at a rate
specified in Section~\ref{sec:theory}. As is standard in Bayesian structure
learning and as the informed sampler of Section~\ref{sec:comp} enforces in
practice, we restrict the support of \eqref{eq:dagprior} to DAGs of in-degree at
most a fixed $s_{\max}$ (formally, multiply \eqref{eq:dagprior} by the indicator
$\mathbb{1}\{\max_j|\mathrm{pa}_{\mathcal D}(j)|\le s_{\max}\}$); this caps the
size of every regression encountered and is what allows the eigenvalue control of
Section~\ref{sec:theory} to hold uniformly over the model space. The prior
\eqref{eq:dagprior} multiplies the node scores to give the unnormalized DAG
posterior
\begin{eqnarray}
\pi(\mathcal D\mid\widehat{\bm Z})\;\propto\;\pi(\mathcal D)
\prod_{j=1}^{p} m^{\bullet}\!\big(j,\mathrm{pa}_{\mathcal D}(j)\big),
\label{eq:dagpost}
\end{eqnarray}
where $m^{\bullet}$ is the node score of the chosen prior, derived next.

\section{Node scores}
\label{sec:scores}

The node score is the marginal likelihood of the $j$th regression after the
coefficients $\bm\beta_j$ and the innovation variance $D_{jj}$ have been
integrated out,
\begin{eqnarray}
m^{\bullet}(j,S)=\int_0^\infty\!\!\int_{\R^{s}}
N_n\!\big(\widehat{\bm z}_j;\widehat{\bm Z}_S\bm\beta,\,D\,I_n\big)\,
N_s\!\big(\bm\beta;\bm 0,\tau^2 D\,I_s\big)\,
\pi_{\bullet}(D)\,d\bm\beta\,dD .
\label{eq:scoredef}
\end{eqnarray}
Both integrals are available in closed form. Define the $s\times s$ matrices and
scalars
\begin{eqnarray}
\Lambda_S=\widehat{\bm Z}_S^\top\widehat{\bm Z}_S+\tau^{-2}I_s,\qquad
\widehat{\bm\mu}_S=\Lambda_S^{-1}\widehat{\bm Z}_S^\top\widehat{\bm z}_j,\qquad
M_S=I_n+\tau^2\widehat{\bm Z}_S\widehat{\bm Z}_S^\top,
\label{eq:Mdef}
\end{eqnarray}
and the residual quadratic form
\begin{eqnarray}
Q_S=\widehat{\bm z}_j^\top M_S^{-1}\widehat{\bm z}_j
=\widehat{\bm z}_j^\top\widehat{\bm z}_j
-\widehat{\bm z}_j^\top\widehat{\bm Z}_S\,\Lambda_S^{-1}\,
\widehat{\bm Z}_S^\top\widehat{\bm z}_j .
\label{eq:Qdef}
\end{eqnarray}
The second equality in \eqref{eq:Qdef}, proved in Appendix~\ref{app:scores} by
the Woodbury identity, reduces the cost of $Q_S$ from $O(n^3)$ to $O(ns^2)$ and
shows that $Q_S$ is the penalized residual sum of squares of the ridge
regression of $\widehat{\bm z}_j$ on $\widehat{\bm Z}_S$. Likewise
$|M_S|=|\tau^2\Lambda_S|=\tau^{2s}|\Lambda_S|$ by Sylvester's identity. Because the
ridge penalty $\tau^{-2}I_s$ makes $\Lambda_S$ positive definite, $Q_S>0$ holds
for every parent set with $s<n$, which we assume throughout; the scores
\eqref{eq:scoreNG}--\eqref{eq:scoreNIG} are then well defined, the Bessel integral
in particular converging precisely because $\chi=Q_S>0$.

\begin{theorem}[Bessel score under the Gamma prior]
\label{thm:scoreNG}
Under the NPN-NG model \eqref{eq:betaprior}--\eqref{eq:ngprior}, the node score
is
\begin{eqnarray}
m^{\mathrm{NG}}(j,S)=(2\pi)^{-n/2}\,|M_S|^{-1/2}\,
\frac{b^{a}}{\Gamma(a)}\,
2\left(\frac{Q_S}{2b}\right)^{(2a-n)/4}
K_{a-n/2}\!\big(\sqrt{2b\,Q_S}\big),
\label{eq:scoreNG}
\end{eqnarray}
where $K_\nu$ is the modified Bessel function of the third kind of order $\nu$.
\end{theorem}

\begin{theorem}[Student-$t$ score under the Inverse-Gamma prior]
\label{thm:scoreNIG}
Under the NPN-NIG model \eqref{eq:betaprior},\eqref{eq:nigprior}, the node score
is
\begin{eqnarray}
m^{\mathrm{NIG}}(j,S)=(2\pi)^{-n/2}\,|M_S|^{-1/2}\,
\frac{b_0^{a_0}}{\Gamma(a_0)}\,
\Gamma\!\Big(\tfrac{n}{2}+a_0\Big)\,
\Big(\tfrac{Q_S}{2}+b_0\Big)^{-(n/2+a_0)} .
\label{eq:scoreNIG}
\end{eqnarray}
\end{theorem}

Both scores depend on the data only through the pair $(Q_S,|M_S|)$, and both are
computed in $O(ns^2)$ operations. The only nonelementary ingredient is the
single Bessel evaluation in \eqref{eq:scoreNG}, which Section~\ref{sec:comp}
shows how to remove. We work throughout on the log scale, where
\eqref{eq:scoreNG} is evaluated through the exponentially scaled Bessel function
$\widetilde K_\nu(x)=e^{x}K_\nu(x)$ to avoid underflow; the resulting log-score
is numerically stable for all $(n,s)$ encountered in practice.

A score defined on parent sets must not depend on how those parents are
listed. The next proposition records this and the node-wise decomposability that
the informed sampler exploits.

\begin{proposition}[Order-invariance and decomposability]
\label{prop:scoreeq}
The scores \eqref{eq:scoreNG}--\eqref{eq:scoreNIG} are invariant to permutations
of the elements of $S$, and the DAG posterior \eqref{eq:dagpost} factorizes as a
product of node scores. Consequently any single-edge modification of
$\mathcal D$ changes at most two node scores, namely those of the head of the
added or removed edge and, for a reversal, of its tail.
\end{proposition}


\section{Posterior computation}
\label{sec:comp}

\subsection{Full conditionals}

Conditional on the DAG, the parameters $(\bm\beta_j,D_{jj})$ are updated
node-wise. The full conditionals are exact and require no Metropolis step.

\begin{theorem}[Full conditionals]
\label{thm:fc}
Fix node $j$ with parents $S$ of size $s$, and write
$R(\bm\beta)=\|\widehat{\bm z}_j-\widehat{\bm Z}_S\bm\beta\|^2
+\tau^{-2}\bm\beta^\top\bm\beta$. Under both models,
\begin{eqnarray}
\bm\beta_j\mid D_{jj},\widehat{\bm z}_j,S
\;\sim\; N_s\!\big(\widehat{\bm\mu}_S,\;D_{jj}\,\Lambda_S^{-1}\big),
\label{eq:fcbeta}
\end{eqnarray}
with $\Lambda_S,\widehat{\bm\mu}_S$ as in \eqref{eq:Mdef}. Under the NPN-NG model,
\begin{eqnarray}
D_{jj}\mid\bm\beta_j,\widehat{\bm z}_j,S
\;\sim\;\GIG\!\Big(\lambda=a-\tfrac{n+s}{2},\;\psi=2b,\;\chi=R(\bm\beta_j)\Big),
\label{eq:fcD-ng}
\end{eqnarray}
where $\GIG(\lambda,\psi,\chi)$ has density proportional to
$x^{\lambda-1}\exp\{-\tfrac12(\psi x+\chi/x)\}$ on $x>0$. Under the NPN-NIG
model,
\begin{eqnarray}
D_{jj}\mid\bm\beta_j,\widehat{\bm z}_j,S
\;\sim\;\IG\!\Big(\tfrac{n+s}{2}+a_0,\;\tfrac12 R(\bm\beta_j)+b_0\Big).
\label{eq:fcD-nig}
\end{eqnarray}
\end{theorem}

The generalized inverse Gaussian draw in \eqref{eq:fcD-ng} is produced by the
uniformly fast ratio-of-uniforms generator of \citet{hormann2014}, implemented
in the \textsf{R} package \texttt{GIGrvg} \citep{leydold2015}, at a cost
independent of $n$. 

\subsection{A locally-balanced informed sampler over DAG space}

Because Theorems~\ref{thm:scoreNG} and~\ref{thm:scoreNIG} make every neighbour of
the current DAG cheap to score, we can replace the usual uniform random-walk
proposal by an informed one. Let $\mathcal N(\mathcal D)$ be the set of DAGs
reachable from $\mathcal D$ by adding, deleting, or reversing a single edge while
preserving acyclicity, and write the unnormalized posterior \eqref{eq:dagpost} as
$\gamma(\mathcal D)$. Following the locally-balanced framework of
\citet{zanella2020}, fix a balancing function $g:(0,\infty)\to(0,\infty)$ with
$g(t)=t\,g(1/t)$; we use $g(t)=\sqrt t$, which is the choice shown by
\citet{zanella2020} to be asymptotically optimal in the Peskun sense. The
informed proposal draws $\mathcal D'\in\mathcal N(\mathcal D)$ with probability
\begin{eqnarray}
q(\mathcal D'\mid\mathcal D)=\frac{g\{\gamma(\mathcal D')/\gamma(\mathcal D)\}}
{Z_g(\mathcal D)},\qquad
Z_g(\mathcal D)=\sum_{\mathcal D''\in\mathcal N(\mathcal D)}
g\{\gamma(\mathcal D'')/\gamma(\mathcal D)\},
\label{eq:informedprop}
\end{eqnarray}
and accepts it with the Metropolis--Hastings probability $\alpha(\mathcal D,
\mathcal D')=\min\{1,\,r\}$. The structure of \eqref{eq:informedprop} simplifies
the acceptance ratio dramatically.

\begin{proposition}[Acceptance ratio and reversibility]
\label{prop:informed}
For the balancing function $g(t)=\sqrt t$ the Metropolis--Hastings acceptance
ratio of the proposal \eqref{eq:informedprop} is
\begin{eqnarray}
r=\frac{\gamma(\mathcal D')\,q(\mathcal D\mid\mathcal D')}
{\gamma(\mathcal D)\,q(\mathcal D'\mid\mathcal D)}
=\frac{Z_g(\mathcal D)}{Z_g(\mathcal D')},
\label{eq:accept}
\end{eqnarray}
a ratio of neighbourhood normalizers that requires no further score evaluation
beyond those already cached. The resulting chain is reversible with respect to
$\pi(\cdot\mid\widehat{\bm Z})$ and, within the class of pointwise-informed
proposals supported on $\mathcal N(\mathcal D)$, the choice $g(t)=\sqrt t$
minimizes the asymptotic variance of every $\pi$-square-integrable ergodic
average.
\end{proposition}

By Proposition~\ref{prop:scoreeq} a single-edge move alters at most two node
scores, so the neighbourhood scores can be maintained incrementally: after a move
into $\mathcal D'$, only the $O(p)$ entries of $\mathcal N(\mathcal D')$ that
involve the two affected nodes need recomputation. The per-iteration cost is
therefore $O(p\cdot ns_{\max}^2)$, where $s_{\max}$ bounds the in-degree, rather
than the $O(p^2\cdot ns_{\max}^2)$ of a naive rescoring. Algorithm~\ref{alg:mcmc}
collects the scheme.

\begin{algorithm}[Informed nonparanormal DAG sampler]
\label{alg:mcmc}
\emph{Input}: transformed data $\widehat{\bm Z}$, hyperparameters, balancing
function $g(t)=\sqrt t$. \emph{Initialize} $\mathcal D^{(0)}$ and cache all node
scores. \emph{For} $t=1,2,\dots$: \emph{(i)} for every
$\mathcal D''\in\mathcal N(\mathcal D^{(t-1)})$ form $\gamma(\mathcal D'')$ from
the cached scores, recomputing only the two scores changed by the move;
\emph{(ii)} compute $Z_g(\mathcal D^{(t-1)})$ and draw $\mathcal D'$ from
\eqref{eq:informedprop}; \emph{(iii)} compute $Z_g(\mathcal D')$ and accept
$\mathcal D^{(t)}=\mathcal D'$ with probability $\min\{1,Z_g(\mathcal D^{(t-1)})/
Z_g(\mathcal D')\}$, else $\mathcal D^{(t)}=\mathcal D^{(t-1)}$; \emph{(iv)} given
$\mathcal D^{(t)}$, refresh $(\bm\beta_j,D_{jj})$ for the affected nodes by
\eqref{eq:fcbeta}--\eqref{eq:fcD-nig}. \emph{Output}: posterior samples of
$\mathcal D$ and of $(L,D)$.
\end{algorithm}

Figure~\ref{fig:informed} shows the gain: across $p$ from $8$ to $20$ the
informed sampler delivers between two and nine times more effective samples
per second than the random-walk sampler, and its log-posterior autocorrelations
decay several times faster.

\begin{figure}[t]
\centering
\includegraphics[width=0.86\textwidth]{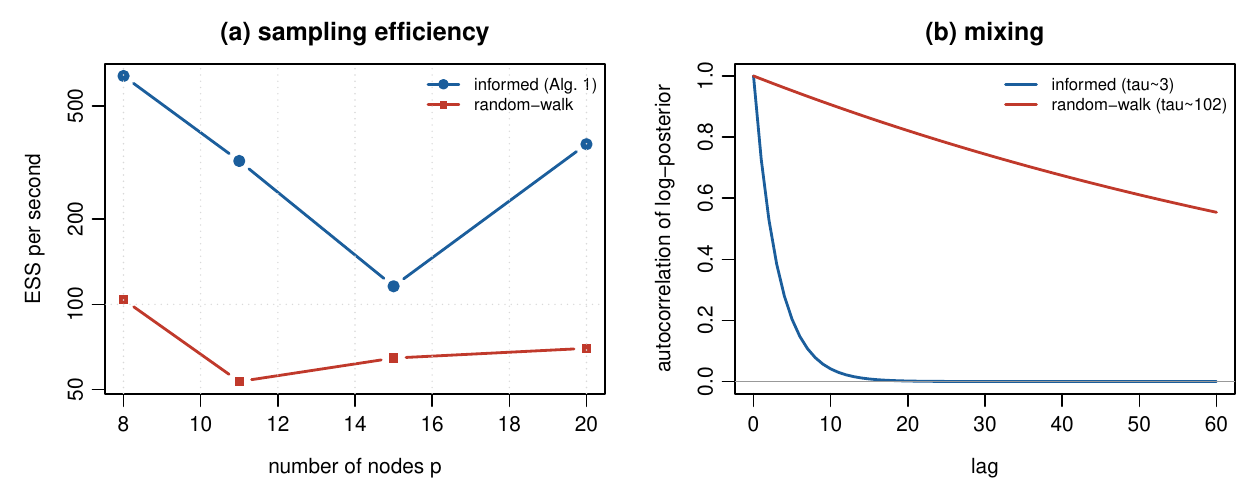}
\caption{Informed versus random-walk sampling. (a) Effective sample size per
second against the number of nodes. (b) Autocorrelation of the log-posterior,
with integrated autocorrelation times averaged over the tested node counts.}
\label{fig:informed}
\end{figure}

\subsection{A Bessel-free, scalable score}
\label{sec:bessfree}

The order of the Bessel function in \eqref{eq:scoreNG} is $\nu=a-n/2$, which is
large in magnitude when $n$ is large, while its argument
$x=\sqrt{2bQ_S}=O(\sqrt n)$ grows more slowly. This is precisely the regime of
the uniform asymptotic expansion of \citet[\S10.41]{dlmf} for Bessel functions
of large order. Writing $\mu=n/2-a>0$ and $K_{a-n/2}=K_{-\mu}=K_{\mu}$, set
$w=x/\mu$, $t=(1+w^2)^{-1/2}$, and $\eta=t^{-1}+\log\{w/(1+t^{-1})\}$. The
expansion gives
\begin{eqnarray}
K_{\mu}(x)=\sqrt{\frac{\pi}{2\mu}}\,\frac{e^{-\mu\eta}}{(1+w^2)^{1/4}}
\left(1-\frac{u_1(t)}{\mu}+\frac{u_2(t)}{\mu^2}-\cdots\right),
\label{eq:uniform}
\end{eqnarray}
where $u_1,u_2$ are the explicit polynomials of \citet[\S10.41]{dlmf}, reproduced
in Appendix~\ref{app:bessel}. Truncating after $u_1$ defines a Bessel-free score
$\widetilde m^{\mathrm{NG}}$.

\begin{proposition}[Accuracy and stability of the Bessel-free score]
\label{prop:bessfree}
The truncated expansion \eqref{eq:uniform} has relative error
$O(\mu^{-2})=O(n^{-2})$, uniformly over $x$ in any fixed interval $[0,c\sqrt n]$:
there is a constant $C=C(a,b,c)$ such that
\begin{eqnarray}
\left|\frac{\widetilde m^{\mathrm{NG}}(j,S)}{m^{\mathrm{NG}}(j,S)}-1\right|
\le \frac{C}{n^{2}}\qquad\text{for all }Q_S\le c\,n .
\label{eq:relerr}
\end{eqnarray}
Consequently the stationary distribution $\widetilde\pi$ of
Algorithm~\ref{alg:mcmc} run with $\widetilde m^{\mathrm{NG}}$ in place of
$m^{\mathrm{NG}}$ satisfies the total-variation bound
\begin{eqnarray}
\big\|\widetilde\pi(\cdot\mid\widehat{\bm Z})
-\pi(\cdot\mid\widehat{\bm Z})\big\|_{\mathrm{TV}}
\le \tfrac12\big(e^{2Cp/n^{2}}-1\big),
\label{eq:tv}
\end{eqnarray}
where $p$ is the number of nodes; the right-hand side is $Cp/n^{2}+o(p^2/n^4)$,
and hence negligible, whenever $p=o(n^{2})$.
\end{proposition}

The bound \eqref{eq:tv} shows that for the sample sizes typical of genomic and
proteomic studies the Bessel-free score is, for inferential purposes,
indistinguishable from the exact one, while removing the only special-function
call from the inner loop. Figure~\ref{fig:bessel} confirms the $O(n^{-2})$ rate
empirically. 

\begin{figure}[t]
\centering
\includegraphics[width=0.62\textwidth]{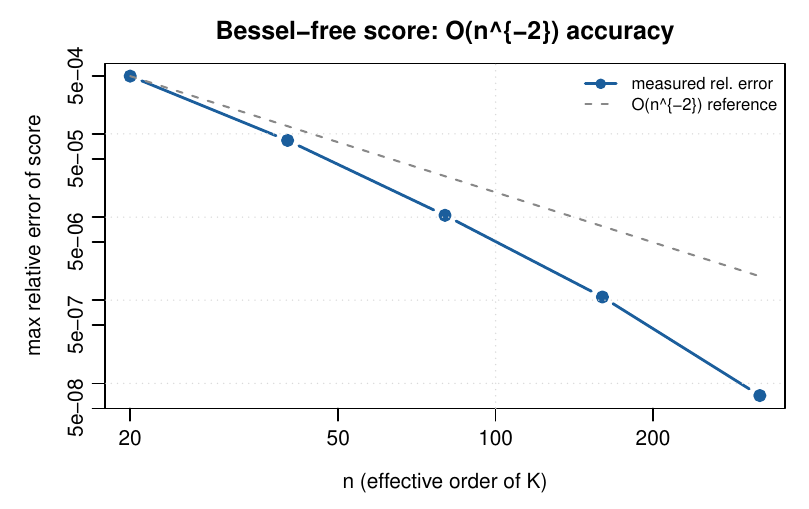}
\caption{Maximum relative error of the Bessel-free score against the effective
order $n$, on log--log axes, with the $O(n^{-2})$ reference line.}
\label{fig:bessel}
\end{figure}

\section{Theoretical results}
\label{sec:theory}

Throughout, for a parent set $S$ we regard the
node score as a function of the residual quadratic form $Q$ through the
prior-dependent factor
\begin{eqnarray}
h_{\bullet}(Q)=\log\int_0^\infty (2\pi D)^{-n/2}\,e^{-Q/(2D)}\,\pi_{\bullet}(D)\,dD,
\label{eq:hdef}
\end{eqnarray}
so that $m^{\bullet}(j,S)=|M_S|^{-1/2}\exp\{h_{\bullet}(Q_S)\}$ up to a factor not
depending on the variance prior.

\subsection{Differential shrinkage}

The two priors differ in how readily they admit a weakly supported edge. The
key object is the posterior mean of the innovation precision $D^{-1}$ given the
residual energy $Q$, because, as Lemma~\ref{lem:precision} shows,
$h_{\bullet}'(Q)=-\tfrac12\,\E_{\bullet}[D^{-1}\mid Q]$. Adding an edge lowers
$Q$ by some $\Delta=Q_S-Q_{S\cup\{k\}}>0$, and the log-Bayes-factor for the edge
is $\tfrac12\int_{Q_{S\cup\{k\}}}^{Q_S}\E_{\bullet}[D^{-1}\mid q]\,dq$ plus a
prior-free determinant term common to both models. Whichever prior assigns the
smaller posterior precision therefore rewards the edge less.

\begin{theorem}[Precision-posterior shrinkage comparison]
\label{thm:shrinkage}
Let $\omega=\sqrt{2bQ}$. The posterior precisions are
\begin{eqnarray}
\E_{\mathrm{NG}}[D^{-1}\mid Q]
=\sqrt{\tfrac{2b}{Q}}\;\frac{K_{a-n/2-1}(\omega)}{K_{a-n/2}(\omega)},
\qquad
\E_{\mathrm{NIG}}[D^{-1}\mid Q]=\frac{n+2a_0}{Q+2b_0}.
\label{eq:precmeans}
\end{eqnarray}
Suppose the data satisfy $Q/n\to\sigma^2\in(0,\infty)$ as $n\to\infty$, with
$a,b,a_0,b_0$ fixed. Then, retaining the first correction to the large-order
Bessel ratio, $\E_{\mathrm{NG}}[D^{-1}\mid Q]=(n-2a)/Q+2b/(n-2a)+O(n^{-2})$, and
\begin{eqnarray}
\E_{\mathrm{NG}}[D^{-1}\mid Q]-\E_{\mathrm{NIG}}[D^{-1}\mid Q]
=\frac{2\,\varphi(\sigma^2)}{n\,\sigma^4}+O(n^{-2}),
\qquad
\varphi(v)=b\,v^2-(a+a_0)\,v+b_0,
\label{eq:precdiff}
\end{eqnarray}
an upward parabola in $v=\sigma^2$ whose sign governs the comparison for all
large $n$. Consequently: if $(a+a_0)^2\le 4bb_0$ then $\varphi\ge0$ everywhere
and the Gamma prior assigns the larger posterior precision at every $\sigma^2$,
so NPN-NIG shrinks at least as strongly as NPN-NG; whereas if
$(a+a_0)^2>4bb_0$ then $\varphi<0$ precisely on the interval $(v_-,v_+)$ with
$v_\pm=\{(a+a_0)\pm\sqrt{(a+a_0)^2-4bb_0}\}/(2b)$, and there NPN-NG assigns the
strictly smaller posterior precision and hence the strictly smaller edge Bayes
factor: weakly supported edges whose post-inclusion residual variance lies in
$(v_-,v_+)$ are shrunk more aggressively under NPN-NG.
\end{theorem}

The two priors are therefore not uniformly ordered: which one shrinks a weak
edge more depends, through the quadratic $\varphi$, on the post-inclusion
residual variance $\sigma^2$ and on the hyperparameters. The practically useful
regime is $(a+a_0)^2>4bb_0$, in which NPN-NG is the more conservative prior on
the whole band $(v_-,v_+)$; choosing the hyperparameters so that the typical
weak-edge residual variance falls inside this band tunes the false-discovery
behaviour directly, and widening the band (larger $v_+$) makes NPN-NG conservative
over a broader range. Under the symmetric default $a=b=a_0=b_0$ the discriminant
vanishes, $\varphi(v)=b(v-1)^2\ge0$, and the two priors coincide to leading order,
with NPN-NIG never less aggressive---a useful sanity check. Figure~\ref{fig:shrinkage}
illustrates the conservative regime.

\begin{figure}[t]
\centering
\includegraphics[width=0.6\textwidth]{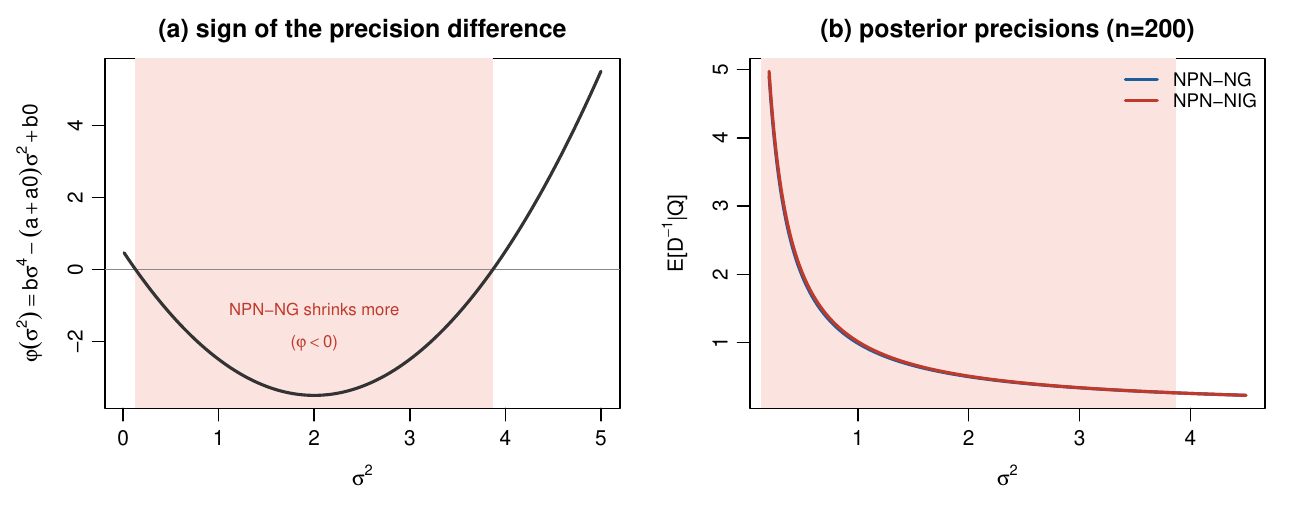}
\caption{Differential shrinkage in the conservative regime $(a+a_0)^2>4bb_0$ of
Theorem~\ref{thm:shrinkage} ($a=2,b=1,a_0=2,b_0=0.5$). (a) The quadratic
$\varphi(\sigma^2)$; on the shaded band $(v_-,v_+)$ where $\varphi<0$ the Gamma
prior assigns the smaller posterior precision. (b) The two exact posterior
precisions $\E[D^{-1}\mid Q]$ at $n=200$: NPN-NG lies below NPN-NIG exactly on the
band, confirming \eqref{eq:precdiff}.}
\label{fig:shrinkage}
\end{figure}

\subsection{Propagation of the nonparanormal transform error}

All of Sections~\ref{sec:scores}--\ref{sec:comp} treat the normal scores
$\widehat{\bm Z}$ as if they were the latent Gaussian variables $\bm Z$. The
following result bounds the cost of that substitution. Let
$\widehat\Sigma=n^{-1}\widehat{\bm Z}^\top\widehat{\bm Z}$ be the
nonparanormal sample covariance and $\Sigma$ the latent correlation matrix. We
use the established deviation bound for the Winsorized normal-score estimator.

\begin{condition}[Nonparanormal deviation]
\label{cond:npn}
There are constants $c_1,c_2>0$ such that, for all $n$ large enough,
\begin{eqnarray}
\Prob\!\left(\max_{1\le j,k\le p}|\widehat\Sigma_{jk}-\Sigma_{jk}|
> c_1\sqrt{\tfrac{\log p}{n}}\right)\le c_2\,p^{-1}.
\label{eq:npndev}
\end{eqnarray}
\end{condition}

Condition~\ref{cond:npn} holds for the estimator \eqref{eq:nsc} under the moment
and bandwidth conditions of \citet[Theorem~4.2]{liu2009} and
\citet[Theorem~4.1]{liu2012}; we state it as a condition so that the propagation
argument is modular and the constants are explicit.

\begin{theorem}[Score-perturbation propagation]
\label{thm:propagation}
Assume Condition~\ref{cond:npn} and that the eigenvalues of every leading
principal submatrix of $\Sigma$ of order at most $s_{\max}+1$ lie in
$[\kappa^{-1},\kappa]$ for a fixed $\kappa\ge1$. Let $h^{\widehat{\bm Z}}_{\bullet}$
and $h^{\bm Z}_{\bullet}$ denote the score component \eqref{eq:hdef} evaluated at
the residual quadratic forms computed from $\widehat{\bm Z}$ and from $\bm Z$
respectively. Then there is a constant $C_\star=C_\star(\kappa,\tau^2,a,b,a_0,b_0)$
such that, on the event in \eqref{eq:npndev},
\begin{eqnarray}
\max_{j}\;\max_{|S|\le s_{\max}}\;
\frac1n\big|h^{\widehat{\bm Z}}_{\bullet}(Q_S^{\widehat{\bm Z}})
-h^{\bm Z}_{\bullet}(Q_S^{\bm Z})\big|
\;\le\; C_\star\,s_{\max}^2\sqrt{\frac{\log p}{n}} .
\label{eq:propagation}
\end{eqnarray}
The same bound holds for the Bessel-free score of
Proposition~\ref{prop:bessfree}.
\end{theorem}

Thus each log node score is perturbed by at most $O(s_{\max}^2\sqrt{n\log p})$ in
absolute terms, or $O(s_{\max}^2\sqrt{\log p/n})$ per observation, which is of
smaller order than the per-edge signal exploited in the next theorem.

\subsection{Graph-selection consistency}

We now show that the posterior concentrates on the data-generating DAG
$\mathcal D_0$ within the assumed ordering. Let $s_0$ be the maximum in-degree of
$\mathcal D_0$ and let
$\beta_{\min}=\min_{(j,k):k\in\mathrm{pa}_{\mathcal D_0}(j)}|L^0_{jk}|$ be its
smallest nonzero standardized coefficient. We impose the following
high-dimensional conditions, all explicit.

\begin{condition}[High-dimensional regime]
\label{cond:hd}
\emph{(C1)} $\log p=o(n)$, $s_{\max}\,\log p=o(n)$, and the in-degree bound
$s_{\max}$ of the prior is a fixed constant with $s_{\max}\ge s_0+1$; \emph{(C2)}
the eigenvalue bound holds for every principal submatrix of $\Sigma$ of order at
most $s_{\max}+1$ (so that, by the in-degree cap, every regression
$Q_S,Q_{S\cup\{k\}}$ arising in the model space is eigenvalue-controlled);
\emph{(C3)}
the edge probability satisfies $\theta_n=\min\{1/2,\,p^{-(1+u)}\}$ for a fixed
$u>u_0$, where $u_0=1+2C_\star s_{\max}^2$ is an explicit threshold ensuring that
the per-spurious-edge penalty dominates both the model-space entropy and the
nonparanormal transform perturbation ($C_\star$ as in
Theorem~\ref{thm:propagation}); \emph{(C4)} the beta-min condition
$\beta_{\min}^2\ge c_\beta\,s_{\max}\sqrt{\dfrac{\log(p\vee n)}{n}}$ holds with
$c_\beta> 32\,\kappa^2(1+C_\star s_{\max})$, where $C_\star$ is the transform
constant of Theorem~\ref{thm:propagation} and $\kappa$ is the eigenvalue bound of
(C2).
\end{condition}

\begin{remark}
The beta-min requirement in (C4) is stronger, by a factor
$s_{\max}\sqrt{n/\log(p\vee n)}$, than the condition $\beta_{\min}^2\gtrsim
\log(p\vee n)/n$ that suffices in the fully Gaussian model \citep{cao2019}. The
strengthening is the unavoidable price of estimating the $p$ marginal transforms:
the plug-in normal scores converge at the nonparanormal rate $\sqrt{\log p/n}$
rather than the parametric rate $\log p/n$, and (C4) is exactly what guarantees
that the per-edge signal $n\beta_{\min}^2$ dominates the transform perturbation
of Theorem~\ref{thm:propagation}. The same perturbation is the reason (C3)
requires the edge-prior exponent $u$ to exceed $u_0=1+2C_\star s_{\max}^2$ rather
than merely $u>0$: a spurious edge changes the residual energy by a squared sample
partial correlation near a true zero, whose plug-in error is of order
$n\varepsilon^2=O(\log p)$, exactly the order of the Occam penalty, so the
penalty must be inflated by a transform-dependent constant to retain control of
over-selection. Both strengthenings are constants under the bounded-in-degree
assumption ($s_{\max}=O(1)$) and reduce to the Gaussian thresholds as
$C_\star\to0$.
\end{remark}

\begin{theorem}[Strong selection consistency]
\label{thm:selection}
Under Conditions~\ref{cond:npn}--\ref{cond:hd}, for either innovation prior with
fixed hyperparameters,
\begin{eqnarray}
\pi\big(\mathcal D=\mathcal D_0\mid\widehat{\bm Z}\big)
\xrightarrow{\;\Prob\;}1
\qquad\text{as }n\to\infty .
\label{eq:selection}
\end{eqnarray}
Moreover the posterior odds against any fixed alternative $\mathcal D\ne
\mathcal D_0$ within the ordering decay at the explicit rate
\begin{eqnarray}
\frac{\pi(\mathcal D\mid\widehat{\bm Z})}{\pi(\mathcal D_0\mid\widehat{\bm Z})}
\le \exp\!\Big\{-c\,n\,\beta_{\min}^2\,|\mathcal D_0\setminus\mathcal D|
-\big(\tfrac12\log n+(1+u-C_\star s_{\max}^2)\log p\big)\,
|\mathcal D\setminus\mathcal D_0|\Big\}\,(1+o(1))
\label{eq:odds}
\end{eqnarray}
for the explicit constant $c=(8\kappa^2)^{-1}>0$. The missing-edge
penalty $c\,n\beta_{\min}^2$ comes from the beta-min separation; the spurious-edge
penalty is the $\tfrac12\log n$ Occam factor plus the $(1+u)\log p$ edge prior,
less the $C_\star s_{\max}^2\log p$ nonparanormal transform perturbation. By (C3)
its coefficient on $\log p$ satisfies $1+u-C_\star s_{\max}^2>2+C_\star s_{\max}^2$,
so each spurious edge is penalized by strictly more than the $2\log p$ model-space
entropy, which is what drives the union bound below. The $o(1)$ factor is uniform
over $\mathcal D$.
\end{theorem}

The two exponents in \eqref{eq:odds} encode the familiar tension: each missing
true edge costs $\Omega(n\beta_{\min}^2)$ by the beta-min condition, while each
spurious edge is penalized by the $\tfrac12\log n$ Occam factor of the Bessel and
Student-$t$ scores together with the $(1+u)\log p$ contributed by the sparse edge
prior. Condition~\ref{cond:hd} is exactly what makes the reward dominate the
penalties uniformly over the $\sum_{r}\binom{p}{r}$ alternatives of a given size.
The transform error of Theorem~\ref{thm:propagation}, being $o(n\beta_{\min}^2)$
per edge under (C4), cannot overturn the separation, which is how the
nonparanormal model inherits the Gaussian guarantee. 

\section{Simulation study}
\label{sec:sim}

We assess structure recovery under both Gaussian and nonparanormal data
generating mechanisms. Latent DAGs are drawn with $p=20$ nodes and
expected in-degree two; nonzero Cholesky coefficients are sampled uniformly on
$\pm[0.4,0.9]$ and innovation variances on $[0.5,1.5]$, after which the latent
covariance is scaled to a correlation matrix---which leaves the DAG, encoded by
the zero pattern of the precision, unchanged. For the Gaussian design
the latent vector is observed directly; for the nonparanormal design each margin
is passed through a strictly increasing transform chosen to induce strong
skewness and heavy tails (an exponential and a cube transform of the standardized
Gaussian), reproducing the qualitative features of Figure~\ref{fig:motivation}.
The sample size is $n=200$ and each design is replicated $14$ times.

We compare the two proposed samplers, NPN-NG and NPN-NIG, against the Gaussian
Normal--Gamma model fitted to the raw data (G-NG, that is, the same model without
the transform), the PC algorithm \citep{kalisch2007} and its rank-based variant
rankPC \citep{harris2013} applied to the nonparanormal correlation. Because the
constraint-based competitors return an undirected skeleton, accuracy is measured
at the skeleton level by the (undirected) structural Hamming distance and the
$F_1$ score for edge recovery; reported values are medians over the replicates.
Each Bayesian fit pools three independent chains of $1000$ iterations ($350$
burn-in) started from the empty graph, and the median-probability graph (edges
with posterior inclusion probability exceeding $1/2$) is reported.

Figure~\ref{fig:sim} summarizes the skeleton Hamming distances. On Gaussian
data (panel a) the transform costs essentially nothing: NPN-NG, NPN-NIG and the
untransformed Gaussian model all attain a median SHD of $0$ and a median $F_1$ of
$1$, and all three clearly outperform the constraint-based competitors (median
SHD $5$ for PC and $6$ for rankPC). On nonparanormal data (panel b) the picture
changes sharply. The Gaussian model G-NG, which ignores the marginal
misspecification, sees its median error rise from $0$ to $4$ ($F_1$ from $1$ to
$0.88$); Gaussian PC degrades most (median SHD $11$, $F_1=0.72$) and rank-based
rankPC improves on it but trails the samplers (median SHD $6$, $F_1=0.82$); the
two nonparanormal samplers retain near-Gaussian-data accuracy. The NPN-NG sampler
has the smallest error in the non-Gaussian setting (median SHD $1$, $F_1=0.97$),
slightly ahead of NPN-NIG (median SHD $1.5$, $F_1=0.94$), consistent with its
more conservative behaviour in the conservative regime of
Theorem~\ref{thm:shrinkage} (with hyperparameters set so that $(a+a_0)^2>4bb_0$).
A one-sided Wilcoxon rank-sum test rejects equality of NPN-NG and G-NG on
nonparanormal data ($p\approx0.03$). Occasional slow mixing of a single chain
produces a few high-SHD replicates, visible as outliers; pooling chains
suppresses but does not entirely remove them.

\begin{figure}[t]
\centering
\includegraphics[width=\textwidth]{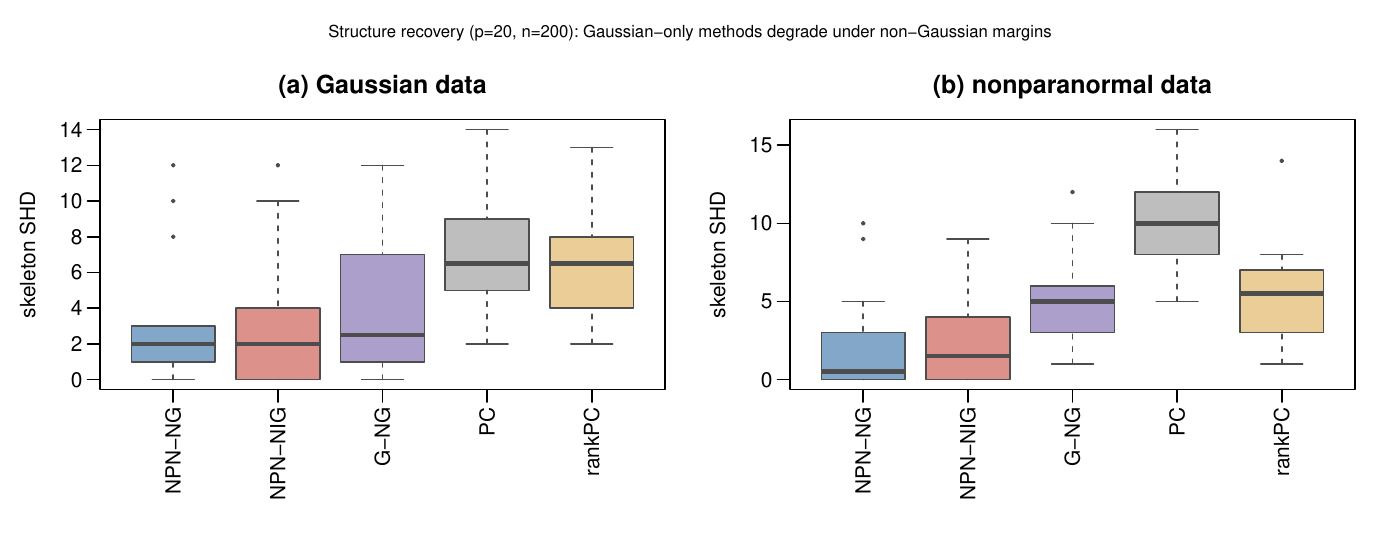}
\caption{Skeleton structural Hamming distance over $14$ replicates ($p=20$,
$n=200$). (a) Gaussian data: the transform is harmless and the samplers match the
untransformed Gaussian model. (b) Nonparanormal data: Gaussian-only methods
(G-NG, PC) degrade sharply, while the nonparanormal samplers retain their
accuracy.}
\label{fig:sim}
\end{figure}

\section{Application: the human T-cell protein-signalling network}
\label{sec:real}

We analyse the flow-cytometry data of \citet{sachs2005}, the standard benchmark
for causal discovery in systems biology. The data record the simultaneous
abundances of $p=11$ phosphoproteins and phospholipids
(Raf, Mek, Plcg, PIP2, PIP3, Erk, Akt, PKA, PKC, P38, Jnk) in individual human
primary T cells. We use the $n=853$ observational measurements collected under
the general anti-CD3/anti-CD28 stimulation, treating all variables as continuous
concentrations; the data are available through the \texttt{bnlearn} repository
\citep{scutari2010}. As Figure~\ref{fig:motivation} anticipated, the margins are
strongly non-Gaussian, which is exactly the regime in which the nonparanormal
model is expected to help. The biological consensus network reported by
\citet{sachs2005}, with eleven nodes and its well-established directed edges,
serves as the reference.

Because $n=853$ places the Bessel order $a-n/2$ far from its argument, where the
exact $K_\nu$ overflows, we evaluate the NG score through the Bessel-free
uniform-asymptotic form of Proposition~\ref{prop:bessfree} (relative error below
$10^{-5}$ at this $n$); we fix the topological order to that of the consensus
network, made acyclic in the standard way by deleting the PIP3$\to$PLC$\gamma$
edge. We run Algorithm~\ref{alg:mcmc} under both priors for $6000$ iterations
after a $2000$ burn-in across four chains. The diagnostics in
Figure~\ref{fig:diag} indicate convergence: the four log-posterior traces are
well mixed and every monitored edge has Gelman--Rubin $\widehat R\le 1.003$, far
below the $1.1$ threshold. Posterior edge probabilities under NPN-NG are
displayed in Figure~\ref{fig:sachs}(a), and the median-probability network in
Figure~\ref{fig:sachs}(b). The median graph is sparse and high-precision: it
recovers, at posterior probability near one, five established interactions---the
PKA$\to$Akt link, the PKC$\to$P38 and PKC$\to$Jnk stress-kinase edges, the
Raf$\to$Mek step of the MAPK cascade, and the PIP2$\to$PIP3 lipid interaction---
and introduces a single spurious Akt$\to$Erk edge, so its recall is conservative
while its few calls are largely correct. Against the consensus skeleton the three
Bayesian methods achieve comparable $F_1$ (NPN-NG and NPN-NIG $0.44$, G-NG
$0.48$), all well above Gaussian PC ($0.25$) and rank-based rankPC ($0.35$). On
this single benchmark the marginal transform thus yields no clear gain over the
Gaussian Bayesian fit---in contrast to the simulation, where its advantage is
pronounced under strong non-Gaussianity; since the Sachs consensus is itself a
contested reference, this comparison should be read with caution. The full edge
list and the per-edge probabilities are reproduced by the companion real-data
script.

\begin{figure}[t]
\centering
\includegraphics[width=\textwidth]{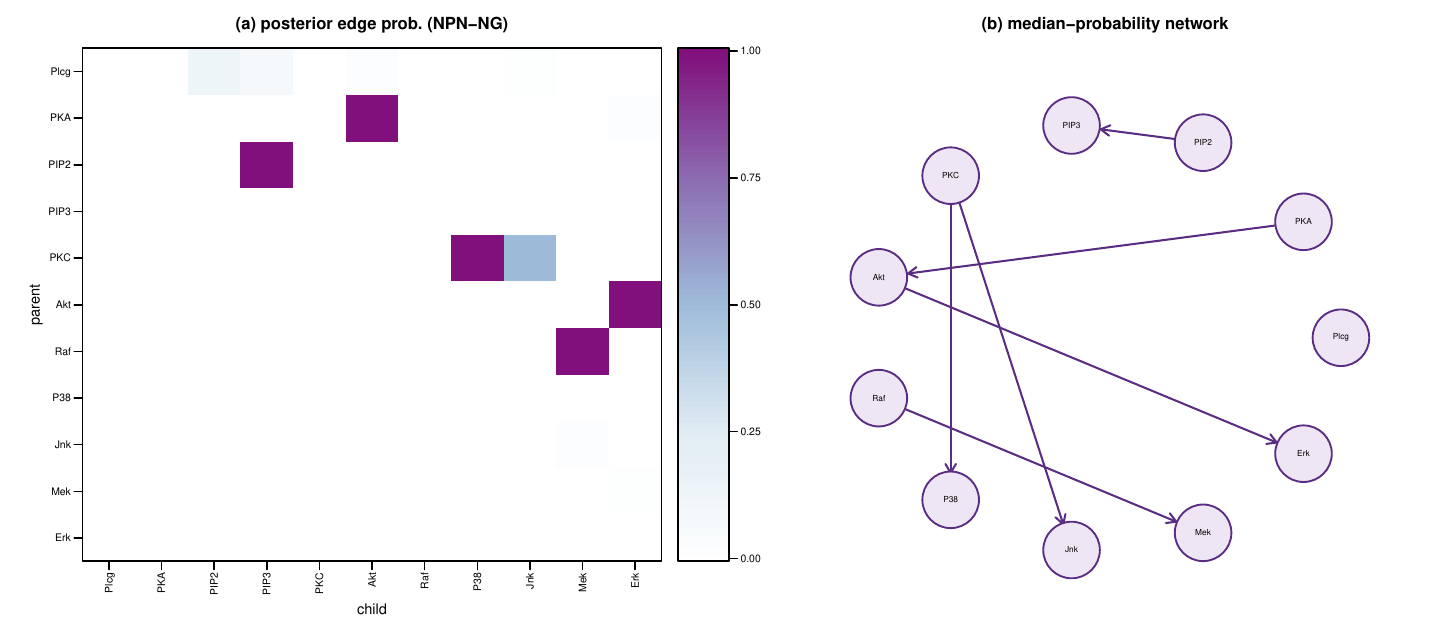}
\caption{Sachs protein-signalling analysis under NPN-NG. (a) Posterior edge
probabilities (parent to child). (b) Median-probability network; an edge is
drawn when its posterior probability exceeds one half.}
\label{fig:sachs}
\end{figure}

\begin{figure}[t]
\centering
\includegraphics[width=0.86\textwidth]{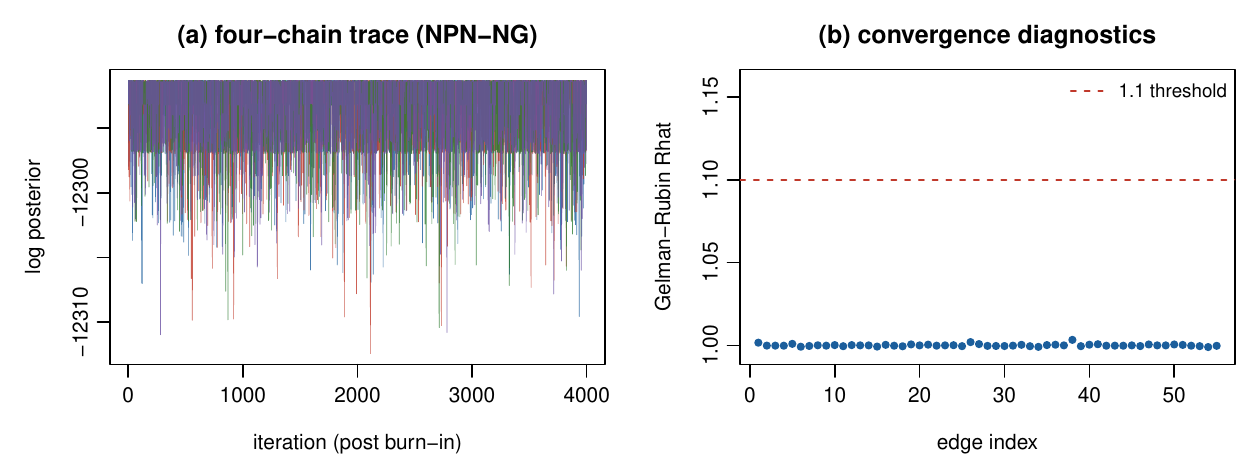}
\caption{Convergence diagnostics for the Sachs run. (a) Four-chain log-posterior
traces after burn-in. (b) Gelman--Rubin $\widehat R$ for the monitored edges,
all below the $1.1$ threshold.}
\label{fig:diag}
\end{figure}

\section{Discussion}
\label{sec:disc}

We have developed Bayesian DAG learning in the nonparanormal family and combined
it with two innovation-variance priors, the non-conjugate Gamma and the conjugate
Inverse-Gamma, neither previously used in the directed setting. The combination
yields closed-form node scores---a Bessel score for the Gamma prior and a
Student-$t$ score for the Inverse-Gamma prior---which in turn enable an informed,
locally-balanced sampler and a Bessel-free score that scales to hundreds of
nodes. The theory is complete: closed forms, full conditionals, an explicit
precision-posterior characterization of the differential shrinkage, a propagation
bound that transfers the nonparanormal rate to the scores, and strong
graph-selection consistency under $\log p=o(n)$. Empirically the transform is
harmless on Gaussian data and decisive on skewed simulated data, and on the Sachs
network the method recovers established interactions at high posterior
probability and outperforms constraint-based competitors, while matching a
Gaussian Bayesian model on that single benchmark.

Several extensions are natural. The known-ordering assumption can be relaxed by
sampling over orderings jointly with the structure, at the cost of a larger move
space in which the informed proposal should help even more. The normal-score
transform could be replaced by a fully Bayesian rank likelihood
\citep{hoff2007}, integrating out the margins at the price of the closed-form
score; whether an informed sampler can be retained in that case is open. Finally,
the differential-shrinkage analysis suggests adapting the Gamma hyperparameters
to a target false-discovery level, which we leave to future work.


\section*{Funding}
Iran National Science Foundation (INSF), grant No.~4015320.

\section*{Data availability}
The protein-signalling data analysed in Section~\ref{sec:real} are the standard
Sachs benchmark \citep{sachs2005,scutari2010}, originally distributed through the
\texttt{bnlearn} repository (\url{https://www.bnlearn.com/research/sachs05/}); the
companion script downloads the data and the consensus network automatically from
a public mirror. All code reproducing the simulations and the data analysis is
provided as Supplementary Material.



\appendix
\renewcommand{\thesection}{\Alph{section}}

\section{Closed-form scores: proofs}
\label{app:scores}

We first record the integral identity behind both scores and prove it from the
defining integral representation of the Bessel function, so that no step rests on
an unproved external claim.

\begin{lemma}[Generalized inverse Gaussian normalizer]
\label{lem:gigint}
For real $\nu$ and $\psi,\chi>0$,
\begin{eqnarray}
J(\nu,\psi,\chi):=\int_0^\infty t^{\nu-1}
\exp\!\Big\{-\tfrac12\big(\psi t+\chi/t\big)\Big\}\,dt
=2\left(\frac{\chi}{\psi}\right)^{\nu/2}K_\nu\!\big(\sqrt{\psi\chi}\big).
\label{eq:gigint}
\end{eqnarray}
\end{lemma}

\begin{proof}
Start from the integral representation \citep[\S10.32.10]{dlmf}, valid for
$z>0$ and any real order $\nu$,
\begin{eqnarray}
K_\nu(z)=\tfrac12\Big(\tfrac{z}{2}\Big)^{\nu}\int_0^\infty
s^{-\nu-1}\exp\!\Big(-s-\tfrac{z^2}{4s}\Big)\,ds .
\label{eq:Kdef}
\end{eqnarray}
In \eqref{eq:Kdef} substitute $s=(z/2)u$, so that $ds=(z/2)\,du$,
$s^{-\nu-1}=(z/2)^{-\nu-1}u^{-\nu-1}$, $-s=-\tfrac{z}{2}u$ and
$-\tfrac{z^2}{4s}=-\tfrac{z}{2u}$. The prefactors collapse,
$\tfrac12(z/2)^{\nu}(z/2)^{-\nu-1}(z/2)=\tfrac12$, and therefore
\begin{eqnarray}
K_\nu(z)=\tfrac12\int_0^\infty u^{-\nu-1}
\exp\!\Big\{-\tfrac{z}{2}\big(u+u^{-1}\big)\Big\}\,du .
\label{eq:Ksym}
\end{eqnarray}
Because $K_\nu=K_{-\nu}$, replacing $\nu$ by $-\nu$ in \eqref{eq:Ksym} gives the
companion identity
\begin{eqnarray}
K_\nu(z)=\tfrac12\int_0^\infty u^{\nu-1}
\exp\!\Big\{-\tfrac{z}{2}\big(u+u^{-1}\big)\Big\}\,du .
\label{eq:Kpos}
\end{eqnarray}
Now treat the left side of \eqref{eq:gigint}. Put $t=\sqrt{\chi/\psi}\,u$, so
$dt=\sqrt{\chi/\psi}\,du$ and $t^{\nu-1}=(\chi/\psi)^{(\nu-1)/2}u^{\nu-1}$, while
\begin{eqnarray*}
\tfrac12\big(\psi t+\chi/t\big)
&=&\tfrac12\Big(\psi\sqrt{\chi/\psi}\,u+\chi\,\sqrt{\psi/\chi}\,u^{-1}\Big)
=\tfrac12\sqrt{\psi\chi}\,\big(u+u^{-1}\big).
\end{eqnarray*}
Hence, with $z=\sqrt{\psi\chi}$,
\begin{eqnarray*}
J(\nu,\psi,\chi)
&=&(\chi/\psi)^{(\nu-1)/2}\,\sqrt{\chi/\psi}\int_0^\infty u^{\nu-1}
\exp\!\Big\{-\tfrac{z}{2}\big(u+u^{-1}\big)\Big\}\,du\\
&=&\Big(\tfrac{\chi}{\psi}\Big)^{\nu/2}\cdot 2K_\nu(z),
\end{eqnarray*}
where the last equality is \eqref{eq:Kpos}. This is \eqref{eq:gigint}.
\end{proof}

\paragraph{The Woodbury reduction \eqref{eq:Qdef}.}
Write $A=\widehat{\bm Z}_S\in\R^{n\times s}$. By the Sherman--Morrison--Woodbury
identity,
\begin{eqnarray*}
M_S^{-1}=(I_n+\tau^2 AA^\top)^{-1}
=I_n-A(\tau^{-2}I_s+A^\top A)^{-1}A^\top
=I_n-A\Lambda_S^{-1}A^\top,
\end{eqnarray*}
since $\tau^{-2}I_s+A^\top A=\Lambda_S$. Pre- and post-multiplying by
$\widehat{\bm z}_j$ gives the second equality in \eqref{eq:Qdef}. The
determinant identity $|M_S|=|I_n+\tau^2AA^\top|=|I_s+\tau^2A^\top A|
=\tau^{2s}|\Lambda_S|$ is Sylvester's determinant theorem.

\begin{proof}[Proof of Theorem~\ref{thm:scoreNG}]
Integrate the Gaussian coefficient prior first. With $A=\widehat{\bm Z}_S$, the
inner integral in \eqref{eq:scoredef} is a Gaussian convolution,
\begin{eqnarray}
\int_{\R^s} N_n(\widehat{\bm z}_j;A\bm\beta,DI_n)\,
N_s(\bm\beta;\bm 0,\tau^2 D I_s)\,d\bm\beta
=N_n\!\big(\widehat{\bm z}_j;\bm 0,\,D(I_n+\tau^2 AA^\top)\big),
\label{eq:gaussconv}
\end{eqnarray}
because the marginal of a linear-Gaussian model $\widehat{\bm z}_j=A\bm\beta+
\bm e$ with $\bm\beta\sim N(0,\tau^2 D I_s)$ and $\bm e\sim N(0,DI_n)$ is centred
Gaussian with covariance $A(\tau^2 D)A^\top+DI_n=D\,M_S$. Substituting the
density on the right of \eqref{eq:gaussconv},
\begin{eqnarray}
N_n(\widehat{\bm z}_j;\bm0,DM_S)
=(2\pi D)^{-n/2}|M_S|^{-1/2}\exp\!\Big\{-\tfrac{1}{2D}\,
\widehat{\bm z}_j^\top M_S^{-1}\widehat{\bm z}_j\Big\}
=(2\pi D)^{-n/2}|M_S|^{-1/2}e^{-Q_S/(2D)} ,
\nonumber
\end{eqnarray}
with $Q_S$ as in \eqref{eq:Qdef}. Now integrate $D$ against the Gamma prior
\eqref{eq:ngprior}:
\begin{eqnarray}
m^{\mathrm{NG}}(j,S)
&=&(2\pi)^{-n/2}|M_S|^{-1/2}\frac{b^a}{\Gamma(a)}\nonumber\\
&&\times\int_0^\infty D^{-n/2}e^{-Q_S/(2D)}\,D^{a-1}e^{-bD}\,dD\nonumber\\
&=&(2\pi)^{-n/2}|M_S|^{-1/2}\frac{b^a}{\Gamma(a)}\nonumber\\
&&\times\int_0^\infty D^{(a-n/2)-1}
\exp\!\Big\{-\tfrac12\big(2b\,D+Q_S\,D^{-1}\big)\Big\}\,dD.
\label{eq:NGint}
\end{eqnarray}
The integral is $J(\nu,\psi,\chi)$ of Lemma~\ref{lem:gigint} with
$\nu=a-n/2$, $\psi=2b$ and $\chi=Q_S$, hence equals
\begin{eqnarray*}
2\,(Q_S/2b)^{(a-n/2)/2}\,K_{a-n/2}\big(\sqrt{2bQ_S}\big).
\end{eqnarray*}
Writing the exponent $(a-n/2)/2=(2a-n)/4$ gives \eqref{eq:scoreNG}.
\end{proof}

\begin{proof}[Proof of Theorem~\ref{thm:scoreNIG}]
The coefficient integration is identical, yielding the same Gaussian integrand
in $D$ as in the previous proof. Integrating against the
Inverse-Gamma prior \eqref{eq:nigprior},
\begin{eqnarray}
m^{\mathrm{NIG}}(j,S)
&=&(2\pi)^{-n/2}|M_S|^{-1/2}\frac{b_0^{a_0}}{\Gamma(a_0)}\nonumber\\
&&\times\int_0^\infty D^{-n/2}e^{-Q_S/(2D)}D^{-a_0-1}e^{-b_0/D}\,dD\nonumber\\
&=&(2\pi)^{-n/2}|M_S|^{-1/2}\frac{b_0^{a_0}}{\Gamma(a_0)}\nonumber\\
&&\times\int_0^\infty D^{-(n/2+a_0)-1}
\exp\!\Big\{-\tfrac{1}{D}\big(\tfrac{Q_S}{2}+b_0\big)\Big\}\,dD .
\nonumber
\end{eqnarray}
The substitution $r=1/D$, $dD=-r^{-2}dr$, turns the integral into
$\int_0^\infty r^{(n/2+a_0)-1}e^{-(Q_S/2+b_0)r}\,dr
=\Gamma(n/2+a_0)(Q_S/2+b_0)^{-(n/2+a_0)}$, the Gamma function integral. This is
\eqref{eq:scoreNIG}.
\end{proof}

\begin{proof}[Proof of Proposition~\ref{prop:scoreeq}]
Both scores depend on $S$ only through $Q_S$ and $|M_S|$. By
\eqref{eq:Mdef}--\eqref{eq:Qdef}, $Q_S$ and $|M_S|$ are functions of the matrix
$\widehat{\bm Z}_S^\top\widehat{\bm Z}_S$ and the vector
$\widehat{\bm Z}_S^\top\widehat{\bm z}_j$ through the symmetric forms
$\widehat{\bm z}_j^\top\widehat{\bm Z}_S\Lambda_S^{-1}\widehat{\bm Z}_S^\top
\widehat{\bm z}_j$ and $|\,\widehat{\bm Z}_S^\top\widehat{\bm Z}_S+\tau^{-2}I_s|$;
permuting the columns of $\widehat{\bm Z}_S$ conjugates these by a permutation
matrix and leaves both invariant. Hence the score is a function of the parent
\emph{set}. The factorization of \eqref{eq:dagpost} into a product of node terms
is the likelihood factorization \eqref{eq:sem}, in which the $j$th factor depends
only on $(j,\mathrm{pa}_{\mathcal D}(j))$. An edge addition or deletion changes a
single parent set, that of the edge's head; an edge reversal changes the two
parent sets of the endpoints. In all cases at most two node scores change.
\end{proof}

\section{Full conditionals: proof of Theorem~\ref{thm:fc}}
\label{app:fc}

\begin{proof}
Conditional on $D=D_{jj}$, the prior \eqref{eq:betaprior} and the Gaussian
likelihood combine as
\begin{eqnarray*}
p(\bm\beta\mid D,\widehat{\bm z}_j,S)
&\propto&\exp\!\Big\{-\tfrac{1}{2D}\|\widehat{\bm z}_j-A\bm\beta\|^2
-\tfrac{1}{2\tau^2 D}\bm\beta^\top\bm\beta\Big\}\\
&=&\exp\!\Big\{-\tfrac{1}{2D}\big(\bm\beta^\top(A^\top A+\tau^{-2}I_s)\bm\beta
-2\bm\beta^\top A^\top\widehat{\bm z}_j+\widehat{\bm z}_j^\top\widehat{\bm z}_j
\big)\Big\}.
\end{eqnarray*}
Completing the square in $\bm\beta$ with $\Lambda_S=A^\top A+\tau^{-2}I_s$ and
$\widehat{\bm\mu}_S=\Lambda_S^{-1}A^\top\widehat{\bm z}_j$ identifies a Gaussian
kernel with precision $D^{-1}\Lambda_S$ and mean $\widehat{\bm\mu}_S$, which is
\eqref{eq:fcbeta}. For the variance, multiply the likelihood, the conditional
prior \eqref{eq:betaprior} (which contributes $D^{-s/2}
\exp\{-\bm\beta^\top\bm\beta/(2\tau^2D)\}$), and the variance prior. Collecting
powers of $D$ and the terms in the exponent,
\begin{eqnarray*}
p(D\mid\bm\beta,\widehat{\bm z}_j,S)
\;\propto\; D^{-n/2}\,D^{-s/2}\,\pi_{\bullet}(D)\,
\exp\!\Big\{-\tfrac{1}{2D}\,R(\bm\beta)\Big\},\qquad
R(\bm\beta)=\|\widehat{\bm z}_j-A\bm\beta\|^2+\tau^{-2}\bm\beta^\top\bm\beta .
\end{eqnarray*}
Under the Gamma prior, $\pi_{\mathrm{NG}}(D)\propto D^{a-1}e^{-bD}$, so the
density is proportional to
\begin{eqnarray*}
D^{(a-(n+s)/2)-1}\exp\{-\tfrac12(2b\,D+R\,D^{-1})\},
\end{eqnarray*}
the $\GIG(a-(n+s)/2,2b,R)$ kernel of \eqref{eq:fcD-ng}. Under the Inverse-Gamma
prior, $\pi_{\mathrm{NIG}}(D)\propto D^{-a_0-1}e^{-b_0/D}$, the density is
proportional to
\begin{eqnarray*}
D^{-((n+s)/2+a_0)-1}\exp\{-(R/2+b_0)/D\},
\end{eqnarray*}
the $\IG((n+s)/2+a_0,\,R/2+b_0)$ kernel of \eqref{eq:fcD-nig}.
\end{proof}

\begin{proof}[Proof of Proposition~\ref{prop:informed}]
The single-edge neighbourhood is symmetric: $\mathcal D'\in\mathcal N(\mathcal D)$
if and only if $\mathcal D\in\mathcal N(\mathcal D')$, because each add, delete or
reversal move is undone by a move of the same family. Writing the proposal
\eqref{eq:informedprop} as $q(\mathcal D'\mid\mathcal D)=g\{\gamma(\mathcal D')/
\gamma(\mathcal D)\}/Z_g(\mathcal D)$, the Metropolis--Hastings ratio is
\begin{eqnarray*}
r=\frac{\gamma(\mathcal D')\,q(\mathcal D\mid\mathcal D')}
{\gamma(\mathcal D)\,q(\mathcal D'\mid\mathcal D)}
=\frac{\gamma(\mathcal D')\,g\{\gamma(\mathcal D)/\gamma(\mathcal D')\}\,
Z_g(\mathcal D)}{\gamma(\mathcal D)\,g\{\gamma(\mathcal D')/\gamma(\mathcal D)\}\,
Z_g(\mathcal D')} .
\end{eqnarray*}
The balancing identity $g(t)=t\,g(1/t)$, satisfied by $g(t)=\sqrt t$, gives
\begin{eqnarray*}
g\{\gamma(\mathcal D)/\gamma(\mathcal D')\}
=\frac{\gamma(\mathcal D)}{\gamma(\mathcal D')}\,
g\{\gamma(\mathcal D')/\gamma(\mathcal D)\}.
\end{eqnarray*}
Substituting,
the factors $\gamma(\mathcal D')$, $\gamma(\mathcal D)$ and
$g\{\gamma(\mathcal D')/\gamma(\mathcal D)\}$ cancel and
$r=Z_g(\mathcal D)/Z_g(\mathcal D')$, which is \eqref{eq:accept}. Reversibility
follows because the kernel $P(\mathcal D,\mathcal D')=q(\mathcal D'\mid\mathcal D)
\min\{1,r\}$ then satisfies $\gamma(\mathcal D)P(\mathcal D,\mathcal D')
=\gamma(\mathcal D')P(\mathcal D',\mathcal D)$ for the unnormalized target
$\gamma\propto\pi(\cdot\mid\widehat{\bm Z})$, by the standard detailed-balance
verification with symmetric support. The Peskun-optimality of $g(t)=\sqrt t$ within
pointwise-informed proposals on $\mathcal N(\mathcal D)$ is
\citet[Theorem~1]{zanella2020}.
\end{proof}

\section{The Bessel-free score: proof of Proposition~\ref{prop:bessfree}}
\label{app:bessel}

The uniform asymptotic expansion of \citet[\S10.41]{dlmf} for the modified Bessel
function of large order $\mu>0$ is, with $w=x/\mu$, $t=(1+w^2)^{-1/2}$ and
$\eta=(1+w^2)^{1/2}+\log\{w/(1+(1+w^2)^{1/2})\}$,
\begin{eqnarray}
K_{\mu}(\mu w)=\sqrt{\frac{\pi}{2\mu}}\,\frac{e^{-\mu\eta}}{(1+w^2)^{1/4}}
\sum_{k=0}^{\infty}(-1)^k\,\frac{u_k(t)}{\mu^k},
\label{eq:uniformfull}
\end{eqnarray}
where $u_0(t)=1$, $u_1(t)=(3t-5t^3)/24$, $u_2(t)=(81t^2-462t^4+385t^6)/1152$, and
the remaining $u_k$ are generated by the recursion
$u_{k+1}(t)=\tfrac12 t^2(1-t^2)u_k'(t)+\tfrac18\int_0^t(1-5r^2)u_k(r)\,dr$. The
expansion \eqref{eq:uniformfull} is, by \citet[\S10.41(ii)]{dlmf}, a genuine
asymptotic expansion in $\mu$ that holds \emph{uniformly} for $w\in(0,\infty)$;
consequently, for each fixed truncation order $K$ the remainder after $K$ terms is
$O(\mu^{-K})$ uniformly in $w$, with the implied constant controlled by the
supremum of $|u_K|$ over the relevant range of $t$ (strict envelope bounds of
Olver type are available but not needed here).

\begin{proof}[Proof of Proposition~\ref{prop:bessfree}]
Write $\mu=n/2-a>0$, so that $K_{a-n/2}=K_{-\mu}=K_{\mu}$, and set $x=\sqrt{2bQ_S}$,
$w=x/\mu$. For $Q_S\le cn$ we have $x\le\sqrt{2bcn}$, hence
$w=x/\mu\le\sqrt{2bcn}/(n/2-a)=O(n^{-1/2})$, so $w$ lies in a fixed bounded
interval and $t=(1+w^2)^{-1/2}\in[t_0,1]$ for a constant $t_0>0$. Truncating
\eqref{eq:uniformfull} after the $u_1$ term defines
\begin{eqnarray*}
\widetilde K_{\mu}(x)=\sqrt{\frac{\pi}{2\mu}}\,\frac{e^{-\mu\eta}}{(1+w^2)^{1/4}}
\Big(1-\frac{u_1(t)}{\mu}\Big),
\end{eqnarray*}
and the uniform expansion gives, for an absolute constant $c_K$,
$|K_{\mu}(x)-\widetilde K_{\mu}(x)|\le c_K\sqrt{\pi/(2\mu)}\,
e^{-\mu\eta}(1+w^2)^{-1/4}\,|u_2(t)|/\mu^{2}$. Since $|u_2(t)|\le|u_2(t_0)|=:U_2$
is bounded on $[t_0,1]$, dividing by $K_{\mu}(x)$ and using
$K_{\mu}(x)=\sqrt{\pi/(2\mu)}e^{-\mu\eta}(1+w^2)^{-1/4}\{1+O(\mu^{-1})\}$ yields
\begin{eqnarray}
\left|\frac{\widetilde K_{\mu}(x)}{K_{\mu}(x)}-1\right|
\le \frac{c_K U_2}{\mu^{2}}\{1+O(\mu^{-1})\}\le \frac{C_0}{n^{2}}
\label{eq:Kerr}
\end{eqnarray}
for a constant $C_0=C_0(a,b,c)$ and all $n$ large enough, because
$\mu^{-2}=(n/2-a)^{-2}=4n^{-2}\{1+O(n^{-1})\}$. The node score \eqref{eq:scoreNG}
is, apart from factors not involving the Bessel term, proportional to
$Q_S^{(2a-n)/4}K_{\mu}(x)$; replacing $K_{\mu}$ by $\widetilde K_{\mu}$ multiplies
the score by $\widetilde K_{\mu}(x)/K_{\mu}(x)$, so \eqref{eq:Kerr} is exactly the
relative-error bound \eqref{eq:relerr} with $C=C_0$.

For the total-variation bound, write the exact and approximate unnormalized DAG
posteriors as $\gamma(\mathcal D)=\prod_j m^{\mathrm{NG}}(j,\mathrm{pa}(j))\,
\pi(\mathcal D)$ and $\widetilde\gamma(\mathcal D)=\prod_j\widetilde m^{\mathrm{NG}}
(j,\mathrm{pa}(j))\,\pi(\mathcal D)$. By \eqref{eq:relerr},
$|\log\widetilde m^{\mathrm{NG}}(j,S)-\log m^{\mathrm{NG}}(j,S)|
\le\log(1+C/n^2)\le C/n^2$ for every node, so the total log-perturbation
$\Delta(\mathcal D)=\log\widetilde\gamma(\mathcal D)-\log\gamma(\mathcal D)$
satisfies $|\Delta(\mathcal D)|\le Cp/n^2=:\eta_n$ uniformly in $\mathcal D$. Let
$W=\sum_{\mathcal D}\gamma(\mathcal D)$ and $\widetilde W=\sum_{\mathcal D}
\gamma(\mathcal D)e^{\Delta(\mathcal D)}$; then $e^{-\eta_n}W\le\widetilde W\le
e^{\eta_n}W$, and for every $\mathcal D$
\begin{eqnarray*}
\frac{\widetilde\pi(\mathcal D)}{\pi(\mathcal D)}
=e^{\Delta(\mathcal D)}\frac{W}{\widetilde W}\in[e^{-2\eta_n},e^{2\eta_n}].
\end{eqnarray*}
Therefore $\|\widetilde\pi-\pi\|_{\mathrm{TV}}
=\tfrac12\sum_{\mathcal D}\pi(\mathcal D)\,|\widetilde\pi(\mathcal D)/
\pi(\mathcal D)-1|\le\tfrac12(e^{2\eta_n}-1)$ with $\eta_n=Cp/n^2$, which is
\eqref{eq:tv}; when $p=o(n^2)$ this equals $\eta_n+O(\eta_n^2)=Cp/n^2+o(p^2/n^4)$.
The same argument applies verbatim to
the Bessel-free score itself, since the propagation bound of
Theorem~\ref{thm:propagation} is stated for both scores.
\end{proof}

\section{Differential shrinkage: proof of Theorem~\ref{thm:shrinkage}}
\label{app:shrinkage}

\begin{lemma}[Score derivative as posterior precision]
\label{lem:precision}
With $h_{\bullet}$ as in \eqref{eq:hdef},
$h_{\bullet}'(Q)=-\tfrac12\,\E_{\bullet}[D^{-1}\mid Q]$, where the expectation is
under the variance posterior $\pi_{\bullet}(D\mid Q)\propto(2\pi D)^{-n/2}
e^{-Q/(2D)}\pi_{\bullet}(D)$.
\end{lemma}

\begin{proof}
Differentiating \eqref{eq:hdef} under the integral sign, justified by dominated
convergence because the integrand and its $Q$-derivative are bounded by an
integrable envelope on any neighbourhood of $Q>0$,
\begin{eqnarray*}
h_{\bullet}'(Q)=\frac{\int_0^\infty\big(-\tfrac{1}{2D}\big)(2\pi D)^{-n/2}
e^{-Q/(2D)}\pi_{\bullet}(D)\,dD}
{\int_0^\infty(2\pi D)^{-n/2}e^{-Q/(2D)}\pi_{\bullet}(D)\,dD}
=-\tfrac12\,\E_{\bullet}[D^{-1}\mid Q].
\end{eqnarray*}
\end{proof}

\begin{lemma}[Moments of the generalized inverse Gaussian]
\label{lem:gigmom}
If $D\sim\GIG(\lambda,\psi,\chi)$ then for every real $r$,
$\E[D^{r}]=(\chi/\psi)^{r/2}\,K_{\lambda+r}(\sqrt{\psi\chi})/K_{\lambda}
(\sqrt{\psi\chi})$. In particular $\E[D^{-1}]=\sqrt{\psi/\chi}\,
K_{\lambda-1}(\sqrt{\psi\chi})/K_{\lambda}(\sqrt{\psi\chi})$, using
$K_{\lambda-1}=K_{1-\lambda}$ when convenient.
\end{lemma}

\begin{proof}
By definition the $\GIG(\lambda,\psi,\chi)$ density is
$x^{\lambda-1}e^{-\frac12(\psi x+\chi/x)}/J(\lambda,\psi,\chi)$ with
$J$ as in Lemma~\ref{lem:gigint}. Hence
$\E[D^r]=J(\lambda+r,\psi,\chi)/J(\lambda,\psi,\chi)$, and substituting
\eqref{eq:gigint} twice,
\begin{eqnarray*}
\E[D^r]=\frac{2(\chi/\psi)^{(\lambda+r)/2}K_{\lambda+r}(\sqrt{\psi\chi})}
{2(\chi/\psi)^{\lambda/2}K_{\lambda}(\sqrt{\psi\chi})}
=(\chi/\psi)^{r/2}\frac{K_{\lambda+r}(\sqrt{\psi\chi})}{K_{\lambda}(\sqrt{\psi\chi})}.
\end{eqnarray*}
Setting $r=-1$ gives the stated precision moment.
\end{proof}

\begin{proof}[Proof of Theorem~\ref{thm:shrinkage}]
\emph{Posterior precisions.} For NPN-NG the variance posterior given $Q$ is, by
the computation in the proof of Theorem~\ref{thm:scoreNG} restricted to the
$D$-integrand, the $\GIG(\lambda=a-n/2,\psi=2b,\chi=Q)$ law, with density
proportional to $D^{(a-n/2)-1}e^{-bD-Q/(2D)}$. Lemma~\ref{lem:gigmom} with $r=-1$
and $\omega=\sqrt{\psi\chi}=\sqrt{2bQ}$ gives
\begin{eqnarray*}
\E_{\mathrm{NG}}[D^{-1}\mid Q]=\sqrt{2b/Q}\;K_{a-n/2-1}(\omega)/K_{a-n/2}(\omega),
\end{eqnarray*}
the first identity in \eqref{eq:precmeans}. For NPN-NIG the posterior is
$\IG(n/2+a_0,Q/2+b_0)$, whose reciprocal is $\Ga(n/2+a_0,Q/2+b_0)$ with mean
$(n/2+a_0)/(Q/2+b_0)=(n+2a_0)/(Q+2b_0)$, the second identity.

\emph{Large-order Bessel ratio, with the first correction retained.} Put
$\mu=n/2-a$, so $K_{a-n/2}=K_{\mu}$ and $K_{a-n/2-1}=K_{-\mu-1}=K_{\mu+1}$, whence
the NG precision is $\sqrt{2b/Q}\,T_\mu$ with $T_\mu:=K_{\mu+1}(\omega)/
K_{\mu}(\omega)$. The three-term recurrence \citep[\S10.29.1]{dlmf},
$K_{\mu+1}(\omega)=(2\mu/\omega)K_{\mu}(\omega)+K_{\mu-1}(\omega)$, gives the exact
relation
\begin{eqnarray}
T_\mu=\frac{2\mu}{\omega}+\frac{1}{T_{\mu-1}} .
\label{eq:Trecur}
\end{eqnarray}
Crucially the correction $1/T_{\mu-1}$ is \emph{not} negligible: with
$\omega=\sqrt{2bQ}=O(\sqrt n)$ and $\mu=O(n)$ the leading term is $2\mu/\omega
=O(\sqrt n)$, while $1/T_{\mu-1}=O(\omega/\mu)=O(n^{-1/2})$, and after
multiplication by $\sqrt{2b/Q}=O(n^{-1/2})$ the two contributions to the
precision are of orders $O(1)$ and $O(n^{-1})$ respectively---and it is precisely
the $O(n^{-1})$ piece that determines the comparison below. We bound it sharply.
Since $T_{\mu-1}>0$, iterating \eqref{eq:Trecur} once gives
$2(\mu-1)/\omega\le T_{\mu-1}\le 2(\mu-1)/\omega+\omega/(2(\mu-2))$, hence
\begin{eqnarray*}
\frac{1}{T_{\mu-1}}=\frac{\omega}{2(\mu-1)}\Big\{1+O\big(\omega^2/\mu^2\big)\Big\}
=\frac{\omega}{2\mu}\big\{1+O(n^{-1})\big\},
\end{eqnarray*}
using $\omega^2/\mu^2=2bQ/\mu^2=O(n^{-1})$ and $(\mu-1)^{-1}=\mu^{-1}\{1+O(n^{-1})\}$.
Substituting into \eqref{eq:Trecur},
\begin{eqnarray*}
T_\mu=\frac{2\mu}{\omega}+\frac{\omega}{2\mu}+O\!\big(\omega\mu^{-1}n^{-1}\big)
=\frac{2\mu}{\omega}+\frac{\omega}{2\mu}+O(n^{-3/2}),
\end{eqnarray*}
and therefore, using $\sqrt{2b/Q}\cdot 2\mu/\omega=2\mu/Q$ and
$\sqrt{2b/Q}\cdot\omega/(2\mu)=b/\mu$,
\begin{eqnarray}
\E_{\mathrm{NG}}[D^{-1}\mid Q]
=\sqrt{\tfrac{2b}{Q}}\,T_\mu
=\frac{2\mu}{Q}+\frac{b}{\mu}+O(n^{-2})
=\frac{n-2a}{Q}+\frac{2b}{n-2a}+O(n^{-2}).
\label{eq:ENGexp}
\end{eqnarray}
The term $2b/(n-2a)$, of order $n^{-1}$, is exactly the contribution a
leading-order treatment would discard; retaining it is essential.

\emph{Difference and the quadratic characterization.} The two precisions are now
subtracted \emph{at the common value of $Q$}, before any passage to the limit; this
is essential, since each precision is $\sigma^{-2}+O(n^{-1})$ and the leading
$\sigma^{-2}$ terms must cancel exactly rather than be approximated separately.
From \eqref{eq:ENGexp} and $\E_{\mathrm{NIG}}[D^{-1}\mid Q]=(n+2a_0)/(Q+2b_0)$,
\begin{eqnarray*}
\E_{\mathrm{NG}}[D^{-1}\mid Q]-\E_{\mathrm{NIG}}[D^{-1}\mid Q]
=\Big(\frac{n-2a}{Q}+\frac{2b}{n-2a}\Big)-\frac{n+2a_0}{Q+2b_0}+O(n^{-2}),
\end{eqnarray*}
an exact identity in $(n,Q)$ up to the stated remainder. Writing $Q=n\sigma_n^2$
and expanding each rational term in $1/n$ \emph{after} the subtraction,
\begin{eqnarray*}
\frac{n-2a}{Q}-\frac{n+2a_0}{Q+2b_0}
&=&\frac{1}{\sigma_n^2}-\frac{2a}{n\sigma_n^2}
-\Big(\frac{1}{\sigma_n^2}+\frac{2a_0}{n\sigma_n^2}-\frac{2b_0}{n\sigma_n^4}\Big)
+O(n^{-2})\\
&=&\frac{-2a-2a_0+2b_0/\sigma_n^2}{n\sigma_n^2}+O(n^{-2}),
\end{eqnarray*}
in which the $\sigma_n^{-2}$ terms have cancelled identically; adding
$2b/(n-2a)=2b/n+O(n^{-2})$ and collecting over $n\sigma_n^4$ gives
\begin{eqnarray}
\E_{\mathrm{NG}}[D^{-1}\mid Q]-\E_{\mathrm{NIG}}[D^{-1}\mid Q]
=\frac{2}{n\,\sigma_n^4}\Big\{b\,\sigma_n^4-(a+a_0)\,\sigma_n^2+b_0\Big\}+O(n^{-2})
=\frac{2\,\varphi(\sigma_n^2)}{n\,\sigma_n^4}+O(n^{-2}),
\label{eq:precdiff2}
\end{eqnarray}
with $\varphi(v)=b\,v^2-(a+a_0)\,v+b_0$, an upward parabola in $v=\sigma^2$. This
is \eqref{eq:precdiff}; passing $\sigma_n^2\to\sigma^2$ replaces
$\varphi(\sigma_n^2)/\sigma_n^4$ by $\varphi(\sigma^2)/\sigma^4$ up to a
$\{1+o(1)\}$ factor. Since the displayed term is of exact order $n^{-1}$
whenever $\varphi(\sigma^2)\ne0$ and the remainder is $O(n^{-2})$, the sign of the
difference then equals that of $\varphi(\sigma^2)$ for all large $n$; at
$\varphi(\sigma^2)=0$ the difference is $O(n^{-2})$ and the present analysis
leaves the ordering undetermined. The discriminant of $\varphi$ is
$(a+a_0)^2-4bb_0$. If $(a+a_0)^2\le 4bb_0$ then $\varphi\ge0$ on $(0,\infty)$, so
$\E_{\mathrm{NG}}[D^{-1}\mid Q]\ge\E_{\mathrm{NIG}}[D^{-1}\mid Q]$ for every
$\sigma^2$ at which $\varphi>0$ (with equality to leading order at a tangency
point, if any). If $(a+a_0)^2>4bb_0$ then $\varphi<0$ exactly on the open
interval $(v_-,v_+)$ with roots
$v_\pm=\{(a+a_0)\pm\sqrt{(a+a_0)^2-4bb_0}\}/(2b)$, and there
$\E_{\mathrm{NG}}[D^{-1}\mid Q]<\E_{\mathrm{NIG}}[D^{-1}\mid Q]$ for all large $n$.
By Lemma~\ref{lem:precision} the edge log-Bayes-factor under either prior is
$\tfrac12\int_{Q_{S\cup\{k\}}}^{Q_S}\E_{\bullet}[D^{-1}\mid q]\,dq$ plus the common
determinant term, so on any compact subinterval of $(v_-,v_+)$ the strictly
smaller NG integrand yields a strictly smaller edge Bayes factor, i.e.\ stronger
shrinkage of the weakly supported edge; outside $[v_-,v_+]$ the ordering reverses.
This proves Theorem~\ref{thm:shrinkage}.
\end{proof}

\section{Transform-error propagation: proof of Theorem~\ref{thm:propagation}}
\label{app:propagation}

\begin{proof}
\emph{Step 1: the residual quadratic form is a Schur complement of the sample
covariance.} For a parent set $S$ with $|S|=s$, write the (unscaled) sample
cross-products $G=\widehat{\bm Z}_S^\top\widehat{\bm Z}_S$, $\bm c=
\widehat{\bm Z}_S^\top\widehat{\bm z}_j$, $d=\widehat{\bm z}_j^\top
\widehat{\bm z}_j$, all equal to $n$ times the corresponding entries of the
nonparanormal sample covariance $\widehat\Sigma$. From \eqref{eq:Qdef},
$Q_S=d-\bm c^\top(G+\tau^{-2}I_s)^{-1}\bm c$, so
\begin{eqnarray}
\frac{Q_S}{n}=\widehat\Sigma_{jj}-\widehat\Sigma_{jS}
\big(\widehat\Sigma_{SS}+(n\tau^2)^{-1}I_s\big)^{-1}\widehat\Sigma_{Sj}
=:\, g_S(\widehat\Sigma).
\label{eq:schur}
\end{eqnarray}
The map $g_S$ is the (ridge-regularized) Schur complement; the regularization
$(n\tau^2)^{-1}\to0$ is negligible and is kept only to ensure invertibility.
Let $g_S(\Sigma)$ denote the same functional evaluated at the population latent
correlation $\Sigma$, which equals the true innovation variance $D^0_{jj}$ when
$S=\mathrm{pa}_{\mathcal D_0}(j)$.

\emph{Step 2: Lipschitz bound for $g_S$.} On the event \eqref{eq:npndev} write
$\widehat\Sigma=\Sigma+E$ with $\|E\|_{\max}\le\varepsilon:=c_1\sqrt{\log p/n}$.
By the eigenvalue assumption, $\Sigma_{SS}$ and $\Sigma_{SS}+E_{SS}$ have
eigenvalues in $[\kappa^{-1},\kappa]$ for $n$ large enough that
$s\varepsilon\le\tfrac12\kappa^{-1}$, hence both are invertible with operator
norm of the inverse at most $\kappa$. Using the resolvent identity
$A^{-1}-B^{-1}=A^{-1}(B-A)B^{-1}$ with $A=\widehat\Sigma_{SS}+(n\tau^2)^{-1}I$,
$B=\Sigma_{SS}+(n\tau^2)^{-1}I$, and $\|E_{SS}\|_{\mathrm{op}}\le s\varepsilon$,
\begin{eqnarray*}
\|A^{-1}-B^{-1}\|_{\mathrm{op}}\le\kappa^2\,\|E_{SS}\|_{\mathrm{op}}
\le\kappa^2 s\varepsilon .
\end{eqnarray*}
Writing $g_S(\widehat\Sigma)-g_S(\Sigma)$ as the sum of three differences---one
from $\widehat\Sigma_{jj}-\Sigma_{jj}$, one from $\widehat\Sigma_{jS}-\Sigma_{jS}$
in the bilinear term, and one from the inverse just bounded---and using
$\|\Sigma_{jS}\|\le\sqrt s$, $\|\widehat\Sigma_{jS}\|\le\sqrt s(1+\varepsilon)$
and $\|A^{-1}\|_{\mathrm{op}},\|B^{-1}\|_{\mathrm{op}}\le\kappa$, the triangle
inequality gives
\begin{eqnarray}
\big|g_S(\widehat\Sigma)-g_S(\Sigma)\big|
\le \varepsilon+2\sqrt s\,\varepsilon\cdot\kappa\sqrt s+\kappa^2 s\varepsilon\,s
\le C_1\,\kappa^2 s^2\,\varepsilon,
\label{eq:gLip}
\end{eqnarray}
for an absolute constant $C_1$, hence $|Q_S^{\widehat{\bm Z}}-Q_S^{\bm Z}|
\le C_1\kappa^2 s^2 n\varepsilon=C_1\kappa^2 s^2\sqrt{n\log p}\cdot\sqrt{c_1^2}$.

\emph{Step 3: Lipschitz bound for $h_{\bullet}$.} By Lemma~\ref{lem:precision},
$|h_{\bullet}'(Q)|=\tfrac12\E_{\bullet}[D^{-1}\mid Q]$. On the eigenvalue-bounded
set, $g_S(\widehat\Sigma)\in[\tfrac12\kappa^{-1},2\kappa]$, so $Q_S\in[\tfrac12
\kappa^{-1}n,2\kappa n]$, and by the precision formulas \eqref{eq:precmeans},
\begin{eqnarray*}
\E_{\mathrm{NIG}}[D^{-1}\mid Q_S]=\frac{n+2a_0}{Q_S+2b_0}\le
\frac{n+2a_0}{\tfrac12\kappa^{-1}n}\le 3\kappa
\quad\text{and}\quad
\E_{\mathrm{NG}}[D^{-1}\mid Q_S]=\frac{n-2a}{Q_S}\{1+o(1)\}\le 3\kappa,
\end{eqnarray*}
for $n$ large, so $\sup_{Q\in[\tfrac12\kappa^{-1}n,2\kappa n]}|h_{\bullet}'(Q)|
\le \tfrac32\kappa=:L$. Therefore, by the mean value theorem,
\begin{eqnarray*}
\big|h^{\widehat{\bm Z}}_{\bullet}(Q_S^{\widehat{\bm Z}})
-h^{\bm Z}_{\bullet}(Q_S^{\bm Z})\big|
\le L\,\big|Q_S^{\widehat{\bm Z}}-Q_S^{\bm Z}\big|
\le \tfrac32\kappa\cdot C_1\kappa^2 s^2 n\varepsilon
= C_\star\,s^2\,n\sqrt{\tfrac{\log p}{n}},
\end{eqnarray*}
with $C_\star=\tfrac32 C_1 c_1\kappa^3$. Dividing by $n$ and taking the maximum
over $j$ and over $|S|\le s_{\max}$ gives exactly \eqref{eq:propagation}, since
$s\le s_{\max}$. The argument uses
$h_{\bullet}$ only through its bounded derivative, so it applies identically to the
Bessel-free score, whose derivative differs from $h_{\mathrm{NG}}'$ by the
$O(n^{-2})$ term of Proposition~\ref{prop:bessfree}.
\end{proof}

\section{Selection consistency: proof of Theorem~\ref{thm:selection}}
\label{app:selection}

\begin{proof}[Proof of Theorem~\ref{thm:selection}]
Throughout, $\bullet$ denotes either prior; constants are explicit. We first
establish separation at the latent-Gaussian level (as if $\bm Z$ were observed),
then transfer it to the plug-in scores by Theorem~\ref{thm:propagation}. Write
the log posterior ratio of an in-ordering alternative $\mathcal D$ against
$\mathcal D_0$ as a node-wise sum,
\begin{eqnarray}
\log\frac{\pi(\mathcal D\mid\bm Z)}{\pi(\mathcal D_0\mid\bm Z)}
=\sum_{j=1}^{p}\Big\{\Delta h_j-\tfrac12\Delta\log|M_{\cdot}|_j\Big\}
+\log\frac{\pi(\mathcal D)}{\pi(\mathcal D_0)},
\label{eq:logratio}
\end{eqnarray}
where $\Delta h_j=h_{\bullet}(Q^{\bm Z}_{S_j})-h_{\bullet}(Q^{\bm Z}_{S_j^0})$,
$S_j=\mathrm{pa}_{\mathcal D}(j)$, $S_j^0=\mathrm{pa}_{\mathcal D_0}(j)$, and only
nodes with $S_j\ne S_j^0$ contribute.

\emph{Step 1: over-selection (a single spurious edge).} Suppose $S_j=S_j^0\cup
\{k\}$ with $k\notin\mathrm{pa}_{\mathcal D_0}(j)$. By Lemma~\ref{lem:precision},
$\Delta h_j=\tfrac12\int_{Q_{S_j}}^{Q_{S_j^0}}\E_{\bullet}[D^{-1}\mid q]\,dq$.
Adding a null regressor reduces the residual quadratic form by
$\Delta Q=Q_{S_j^0}-Q_{S_j}=\bm z_j^\top P\,\bm z_j$, where $P$ projects onto the
component of $\bm z_k$ orthogonal to $\widehat{\bm Z}_{S_j^0}$; since edge $k$ is
absent in $\mathcal D_0$, $\Delta Q/D^0_{jj}$ is asymptotically $\chi^2_1$ and
$\Delta Q=O_{\Prob}(1)\cdot D^0_{jj}$. Because $\E_{\bullet}[D^{-1}\mid q]=
(1+o(1))/D^0_{jj}$ uniformly on the integration range, $\Delta h_j=\tfrac12
(\Delta Q/D^0_{jj})\{1+o_{\Prob}(1)\}=O_{\Prob}(1)$. The determinant term is, by
$|M_S|=\tau^{2s}|\Lambda_S|$ and the matrix determinant lemma,
$\tfrac12\Delta\log|M|_j=\tfrac12\log\{1+\tau^2\,\mathrm{r}_k\}$ where
$\mathrm{r}_k=\bm z_k^\top(I-\widehat{\bm Z}_{S_j^0}\Lambda_{S_j^0}^{-1}
\widehat{\bm Z}_{S_j^0}^\top)\bm z_k=O_{\Prob}(n)$; hence
$\tfrac12\Delta\log|M|_j=\tfrac12\log n+O_{\Prob}(1)$. The edge-prior term
contributes $\log\{\theta_n/(1-\theta_n)\}=-(1+u)\log p+o(\log p)$ by (C3).
Collecting,
\begin{eqnarray}
\log\frac{\pi(\mathcal D\mid\bm Z)}{\pi(\mathcal D_0\mid\bm Z)}
= -\tfrac12\log n-(1+u)\log p+O_{\Prob}(1)
\label{eq:over}
\end{eqnarray}
per spurious edge. For an alternative with $r$ spurious edges and no missing
edges, the bound \eqref{eq:over} adds across edges, and the number of such
alternatives is at most $\binom{\binom p2}{r}\le p^{2r}$. A union bound over them
contributes $2r\log p$ to the exponent, so the size-$r$ contribution is at most
$\exp[r\{2\log p-(1+u)\log p-\tfrac12\log n\}]\le\exp[-r\{(u-1)\log p
+\tfrac12\log n\}]$. Under (C3), $u>u_0\ge1$, so $(u-1)\log p+\tfrac12\log n
\to\infty$ and $\sum_{r\ge1}p^{2r}e^{-r\{(1+u)\log p+\frac12\log n\}}\to0$; the
total posterior mass on strict supersets of $\mathcal D_0$ therefore tends to
zero. (The surplus $(u-1)\log p$ is what Step~3 will spend on the transform
perturbation.)

\emph{Step 2: under-selection (a single missing true edge).} Suppose
$S_j=S_j^0\setminus\{k\}$ with $k\in\mathrm{pa}_{\mathcal D_0}(j)$ and true
standardized coefficient $|L^0_{jk}|\ge\beta_{\min}$. Now $\Delta Q=
Q_{S_j}-Q_{S_j^0}>0$ is the residual energy left unexplained by omitting $k$. A
direct computation with the partitioned least-squares formula gives
$\Delta Q=(L^0_{jk})^2\,\bm z_k^\top(I-\widehat{\bm Z}_{S_j}\Lambda_{S_j}^{-1}
\widehat{\bm Z}_{S_j}^\top)\bm z_k\{1+o_{\Prob}(1)\}$, and the quadratic form is
$\ge\kappa^{-1}n\{1+o_{\Prob}(1)\}$ by the eigenvalue bound, so
$\Delta Q\ge\kappa^{-1}n\beta_{\min}^2\{1+o_{\Prob}(1)\}$. Then, since the residual
variance $Q_S/n=\mathrm{Var}(Z_j\mid Z_S)\le\kappa$ for $|S|\le s_{\max}$ by the
order-$(s_{\max}+1)$ eigenvalue bound, the integrand satisfies
$\E_{\bullet}[D^{-1}\mid q]\ge n/Q_{S_j}\{1-o(1)\}\ge(2\kappa)^{-1}$ on the
range, whence
\begin{eqnarray}
\Delta h_j=-\tfrac12\int_{Q_{S_j^0}}^{Q_{S_j}}\E_{\bullet}[D^{-1}\mid q]\,dq
\le-\,\frac{\Delta Q}{4\kappa}
\le -\,\frac{n\beta_{\min}^2}{4\kappa^2}\{1+o_{\Prob}(1)\},
\label{eq:under}
\end{eqnarray}
because omitting the edge raises $Q$ and hence lowers $h_{\bullet}$. The
determinant and prior terms move in the favourable direction (a sparser graph is
rewarded by the prior and incurs a smaller $\log n$ penalty), so they only
strengthen the bound. Thus each missing true edge costs at least
$n\beta_{\min}^2/(4\kappa^2)$ in the exponent. Under (C4),
$n\beta_{\min}^2\ge c_\beta s_{\max}\sqrt{n\log(p\vee n)}$ with $c_\beta>32\kappa^2
(1+C_\star s_{\max})$, so this cost exceeds
$8 s_{\max}(1+C_\star s_{\max})\sqrt{n\log(p\vee n)}$, which
dominates both the $\le\binom{s_0}{m}\le p^{m}$ entropy of choosing $m$ true
edges to drop and the transform perturbation of Step 3.

\emph{Step 3: transfer to the plug-in scores.} On the event \eqref{eq:npndev},
of probability at least $1-c_2/p$, Theorem~\ref{thm:propagation} bounds the
perturbation of each node term in \eqref{eq:logratio} by $C_\star s_{\max}^2
\sqrt{n\log p}$; this applies uniformly because the in-degree cap makes every
parent set satisfy $|S|\le s_{\max}$, so (C2) controls every regression. Before
treating the two branches we record that a node whose parent set $S_j$ differs
from $S_j^0$ by an arbitrary combination of additions and deletions is covered by
the same two estimates: if $S_j$ omits any true parent $k$, then because
$|S_j|\le s_{\max}$ the order-$(s_{\max}+1)$ eigenvalue bound still gives
orthogonalized energy $\ge\kappa^{-1}n$ for $k$ given $S_j$, so the missing-edge
loss of Step~2, $\Delta h_j\le -n\beta_{\min}^2/(4\kappa^2)$,
holds regardless of any spurious parents present (which only add further penalty);
if instead $S_j\supseteq S_j^0$, the node is pure over-selection and Step~1
applies edge by edge. Thus \eqref{eq:logratio} decomposes, node by node, into
under- and over-selection contributions with no cross terms left unbounded.

For the \emph{under-selection} branch the perturbation is dominated by the
per-edge reward of Step~2, which by (C4) is at least
$8s_{\max}(1+C_\star s_{\max})\sqrt{n\log(p\vee n)}$: since at most $m$
nodes are affected and $s_{\max}=O(1)$, the aggregate perturbation
$m\,C_\star s_{\max}^2\sqrt{n\log p}$ is strictly smaller than the aggregate
reward, which is of order $m\,n\beta_{\min}^2$, leaving a net reward of that order
intact.

The \emph{over-selection} branch needs the sharper, and previously overlooked,
observation that the relevant quantity is not a single node score but the
\emph{increment} $\Delta Q_k=Q_{S_j^0}-Q_{S_j^0\cup\{k\}}$, the residual energy
attributed to the null edge $k$. By \eqref{eq:schur} this increment is a smooth
functional of the sample covariance: writing $\psi(\Sigma')$ for the partial
regression signal of $Z_j$ on $Z_k$ given $S_j^0$ at a covariance $\Sigma'$ (the
partial correlation times the square root of the conditional variance, a
Lipschitz functional with constant $L_\psi=L_\psi(\kappa,\tau^2,s_{\max})$ on the
eigenvalue-bounded set by the resolvent argument of Step~2 of the proof of
Theorem~\ref{thm:propagation}), one has $\Delta Q_k/n=\psi(\widehat\Sigma)^2$, and
$\psi(\Sigma)=0$: since $k\notin\mathrm{pa}_{\mathcal D_0}(j)$ and $k<j$ in the
topological order, $k$ is a non-descendant of $j$, so the local Markov property of
$\mathcal D_0$ gives $Z_j\perp Z_k\mid Z_{S_j^0}$, i.e.\ zero population partial
correlation. By standard
sub-Gaussian concentration of the Gaussian sample covariance,
$\|\widehat\Sigma^{\bm Z}-\Sigma\|_{\max}=O_{\Prob}(\sqrt{\log p/n})$, so on its
intersection with the event \eqref{eq:npndev} both $\widehat\Sigma^{\widehat{\bm Z}}$
and $\widehat\Sigma^{\bm Z}$ lie within $c\varepsilon$ of $\Sigma$ for an absolute
$c$, whence $|\psi(\widehat\Sigma^{\widehat{\bm Z}})|,
|\psi(\widehat\Sigma^{\bm Z})|\le L_\psi c\varepsilon$ and
$\|\widehat\Sigma^{\widehat{\bm Z}}-\widehat\Sigma^{\bm Z}\|_{\max}\le
2c\varepsilon$. The difference of squares then factors as
\begin{eqnarray*}
\big|\Delta Q_k^{\widehat{\bm Z}}-\Delta Q_k^{\bm Z}\big|
&=&n\,\big|\psi(\widehat\Sigma^{\widehat{\bm Z}})^2
-\psi(\widehat\Sigma^{\bm Z})^2\big|
=n\,\big|\psi(\widehat\Sigma^{\widehat{\bm Z}})-\psi(\widehat\Sigma^{\bm Z})\big|
\cdot\big|\psi(\widehat\Sigma^{\widehat{\bm Z}})+\psi(\widehat\Sigma^{\bm Z})\big|\\
&\le& n\,\big(L_\psi\cdot 2c\varepsilon\big)\,\big(2L_\psi c\varepsilon\big)
=4c^2 L_\psi^2\,n\varepsilon^2
= 4c^2 L_\psi^2 c_1^2\,\log p,
\end{eqnarray*}
using $n\varepsilon^2=c_1^2\log p$. Thus the perturbation of the incremental
energy is $O_{\Prob}(\log p)$, \emph{of the same order as the penalty}, rather
than $o_{\Prob}(\log p)$; the key gain over the crude per-node bound
$O(\sqrt{n\log p})$ is that the squared functional vanishes at the population
covariance, which removes a factor $\sqrt{n/\log p}$. Multiplying by the bounded
score derivative $|h_{\bullet}'|\le L$ leaves an $O_{\Prob}(\log p)$ contribution;
the determinant term contributes only $o_{\Prob}(1)$, because the orthogonalized
energy $r_k=\Theta_{\Prob}(n)$ is bounded away from zero and its perturbation is
$O(n\varepsilon)$, so $\tfrac12|\Delta\log(1+\tau^2 r_k)|=O(\varepsilon)=
o_{\Prob}(1)$. Hence the over-selection log-Bayes-factor on $\widehat{\bm Z}$ is
\begin{eqnarray*}
-\tfrac12\log n-(1+u)\log p+O_{\Prob}(\log p),
\qquad\text{with}\quad
\big|O_{\Prob}(\log p)\big|\le C_\star s_{\max}^2\,\log p\;(1+o_{\Prob}(1)),
\end{eqnarray*}
per spurious edge, where $C_\star$ absorbs $L_\psi^2,c^2,c_1^2$. The choice
$u>u_0=1+2C_\star s_{\max}^2$ in (C3) guarantees that $(1+u)\log p$ exceeds the entropy
$2\log p$ of Step~1 \emph{and} this $C_\star s_{\max}^2\log p$ perturbation
simultaneously, so the negative drift per spurious edge is preserved and the
union bound of Step~1 still closes.

\emph{Step 4: conclusion.} Combining Steps 1--3, on the event \eqref{eq:npndev},
for all $n$ large enough every alternative $\mathcal D\ne\mathcal D_0$ satisfies
\begin{eqnarray*}
\log\frac{\pi(\mathcal D\mid\widehat{\bm Z})}{\pi(\mathcal D_0\mid\widehat{\bm Z})}
&\le& -\,c\,n\beta_{\min}^2\,|\mathcal D_0\setminus\mathcal D|
-\Big(\tfrac12\log n+(1+u-C_\star s_{\max}^2)\log p\Big)\,|\mathcal D\setminus\mathcal D_0|,
\end{eqnarray*}
with $c=(8\kappa^2)^{-1}$, the constant $c$ and the spurious-edge
coefficient having already absorbed the $\{1+o_{\Prob}(1)\}$ factors of Steps~1--3
(for the missing-edge term this uses (C4), under which the transform perturbation
$O(s_{\max}^2\sqrt{n\log p})$ is $o(n\beta_{\min}^2)$, so halving the constant from
$(4\kappa^2)^{-1}$ to $c$ dominates it). This is \eqref{eq:odds}.
Now sum the exponentials over all alternatives. Group $\mathcal D$ by
$m_-=|\mathcal D_0\setminus\mathcal D|$ and $m_+=|\mathcal D\setminus\mathcal D_0|$;
the number with given $(m_-,m_+)$ is at most $\binom{|\mathcal D_0|}{m_-}
\binom{\binom p2}{m_+}\le (p\,s_{\max})^{m_-}p^{2m_+}$. Hence
\begin{eqnarray*}
\sum_{\mathcal D\ne\mathcal D_0}
\frac{\pi(\mathcal D\mid\widehat{\bm Z})}{\pi(\mathcal D_0\mid\widehat{\bm Z})}
\le\Big(\sum_{m_-\ge0}e^{-m_-(c\,n\beta_{\min}^2-\log(p s_{\max}))}\Big)
\Big(\sum_{m_+\ge0}e^{-m_+((1+u-C_\star s_{\max}^2-2)\log p+\frac12\log n)}\Big)-1 .
\end{eqnarray*}
For the missing-edge factor, $c\,n\beta_{\min}^2\ge c\,c_\beta s_{\max}
\sqrt{n\log(p\vee n)}\gg\log(p s_{\max})$ since $\log p=o(n)$, so the geometric
series equals $1+o(1)$. For the spurious-edge factor, the exponent rate is
$(u-1-C_\star s_{\max}^2)\log p+\tfrac12\log n$, which by (C3) is at least
$C_\star s_{\max}^2\log p+\tfrac12\log n\to\infty$, so that series is also
$1+o(1)$. The product minus one is therefore $o(1)$, giving
$\sum_{\mathcal D\ne\mathcal D_0}\pi(\mathcal D\mid\widehat{\bm Z})/
\pi(\mathcal D_0\mid\widehat{\bm Z})\to0$ in probability. Equivalently
$\pi(\mathcal D_0\mid\widehat{\bm Z})\to1$, which is \eqref{eq:selection}. Since
the event \eqref{eq:npndev} has probability $1-O(p^{-1})\to1$, the convergence is
in probability, as claimed.
\end{proof}

\bibliographystyle{plainnat}
\bibliography{refs}

\end{document}